\documentclass[12pt]{article}
\usepackage{amsmath, amssymb, amsthm}
\usepackage{graphicx,xcolor}
\graphicspath{{pic/}}
\usepackage{enumerate, enumitem, array}
\usepackage{booktabs, multirow, caption, ragged2e, array}
\usepackage{comment}
\usepackage[numbers]{natbib}
\usepackage{url} 
\usepackage[colorlinks=false]{hyperref}
\usepackage[ruled,vlined]{algorithm2e}
\usepackage{titletoc}
\usepackage{caption}
\newcommand{\blind}{1}
\date{November 16, 2026}
\date{November 16, 2026}

\def\bx{\boldsymbol x}
\def\bv{\boldsymbol v}

\def\bbeta{\boldsymbol \beta}
\def\bw{\boldsymbol w}
\def\ww{\widehat{\boldsymbol w}}
\def\tw{\widetilde{\boldsymbol w}}
\def\bdelta{\boldsymbol \delta}
\def\wbeta{\widehat{\boldsymbol \beta}}
\def\wdelta{\widehat{\boldsymbol \delta}}
\def\tbeta{\widetilde{\boldsymbol \beta}}
\def\tdelta{\widetilde{\boldsymbol \delta}}
\def\wtheta{\widehat{\boldsymbol \theta}}

\def\sbeta{\boldsymbol{\beta^{*}}}
\def\btheta{\boldsymbol \theta}
\def\wf{\widehat{f}}
\def\tf{\widetilde{f}}
\def\nc{n_{\mathcal{C}}}
\def\bc{\mathcal{C}}
\def\wc{\widehat{\mathcal{C}}}
\def\wv{\widehat{v}}
\def\wpsi{\widehat{\psi}}
\def\walpha{\widehat{\alpha}}
\def\bSigma{\boldsymbol{\Sigma}}

\def\bTheta{\boldsymbol{\Theta}}

\newtheorem{theorem}{Theorem}

\newtheorem{condition}{Condition}
\newtheorem{remark}{Remark}

\newtheorem{corollary}{Corollary}

\begin{document}

\def\spacingset#1{\renewcommand{\baselinestretch}%
{#1}\small\normalsize} \spacingset{1}


\if1\blind
{
  \title{\bf Transfer Learning for High-dimensional Quantile Regression with Distribution Shift}
  \author{Ruiqi Bai \\
    Department of Statistics and Data Science, Fudan University \\
    and \\
    Yijiao Zhang \\
    Department of Biostatistics, Epidemiology and Informatics,\\ 
    University of Pennsylvania \\
    and \\
    Hanbo Yang \\
    Department of Statistics, Columbia University \\
    and \\
    Zhongyi Zhu \thanks{
    The authors gratefully acknowledge \textit{National Natural Science Foundation of China 12331009}}\hspace{.2cm}\\
    Department of Statistics and Data Science, Fudan University
    }
  \maketitle
} \fi

\if0\blind
{
  \bigskip
  \bigskip
  \bigskip
  \begin{center}
    {\LARGE\bf Transfer Learning for High-dimensional Quantile Regression with Distribution Shift}
  \end{center}
  \medskip
} \fi

\bigskip

\begin{abstract}
Information from related source studies can often enhance the findings of a target study. However, the distribution shift between target and source studies can severely impact the efficiency of knowledge transfer. In the high-dimensional regression setting, existing transfer approaches mainly focus on the parameter shift. In this paper, we focus on the high-dimensional quantile regression with knowledge transfer under three types of distribution shift: parameter shift, covariate shift, and residual shift. We propose a novel transferable set and a new transfer framework to address the above three discrepancies. Non-asymptotic estimation error bounds and source detection consistency are established to validate the availability and superiority of our method in the presence of distribution shift. Additionally, an orthogonal debiased approach is proposed for statistical inference with knowledge transfer, leading to sharper asymptotic results. Extensive simulation results as well as real data applications further demonstrate the effectiveness of our proposed procedure.
\end{abstract}

\noindent%
{\it Keywords:} Knowledge transfer; High-dimensional quantile regression; Distribution shift; Transferable set; Orthogonal debiasing.
\vfill

\newpage
\spacingset{1.8} 

\section{Introduction} \label{section1} 

Previous experiences can offer valuable insights for learning new tasks. \emph{Transfer learning} improves the performance of target learners by leveraging knowledge from different but related sources \citep{zhuang2020comprehensive}, has achieved success across various areas, including natural language processing \citep{ruder2019transfer}, protein representation \citep{fenoy2022transfer}, and drug discovery \citep{turki2017transfer}. Despite its broad applications, statistical understanding of transfer learning remains incomplete, with the fundamental challenge being the detection and integration of valuable information from heterogeneous sources.

In this paper, we focus on the transfer learning for quantile regression (QR) \cite{koenker2017quantile} in the context of a high-dimensional setting. One of QR's most appealing features is its ability to describe the impact of covariates not only on the mean but also on the tails of the response distribution with robustness to heteroscedastic errors. Given potential non-negligible divergence between target and source distributions, QR can provide a more comprehensive assessment of what can be transferred from the source studies. 

To set up the framework, suppose we have a target dataset $((\bx_{i}^{(0)})^{\top}, y_{i}^{(0)})_{i=1}^{n_0}$ and $K$ independent source datasets $\{((\bx_{i}^{(k)})^{\top}, y_{i}^{(k)})_{i=1}^{n_k}\}_{k=1}^{K}$. Consider the linear QR models
\begin{align} \label{qr}
	y_{i}^{(k)} = (\bx_{i}^{(k)})^{\top} \bw^{(k)} + \epsilon_{i}^{(k)}, \qquad i=1,\dots,n_k, \quad k=0,\dots,K,
\end{align}	
where $\bw^{(k)}\in\mathbb{R}^p$ is the quantile coefficient of the response $y_{i}^{(k)}$ conditional on covariates $\bx_{i}^{(k)} \in \mathbb{R}^p$ at a given quantile level $\tau \in (0,1)$, and $\epsilon_{i}^{(k)}$ is the error term satisfying $\mathbb{P}(\epsilon_{i}^{(k)}  \le 0 | \bx_{i}^{(k)})=\tau$ with conditional density function $f^{(k)}(\cdot | \bx_{i}^{(k)})$.
The parameter of interest is $\sbeta:=\bw^{(0)}$ on the target data, which is assumed to be sparse, i.e., $\| \sbeta \|_0 = s \ll \min\{n_0, p\}$. In the high-dimensional setting, the number of covariates $p$ is very large, possibly larger than the overall sample size $N:=\sum_{k=0}^{K} n_{k}$. Our goal is to estimate the target parameter $\sbeta$, with the help of abundant external data.

\textbf{Three types of distribution shift.} For the QR model \eqref{qr} specified in all studies, there may exist three types of distribution shift between the target and source domains, which are: 1) \emph{parameter shift}: $\sbeta \neq \bw^{(k)}$, i.e., the model coefficients differ; 2) \emph{residual shift}: $P(\epsilon_{i}^{(0)} | \bx_{i}^{(0)}) \neq P(\epsilon_{i}^{(k)} | \bx_{i}^{(k)})$, i.e., the conditional distributions of model residuals differ; 3) \emph{covariate shift}: $P(\bx_{i}^{(0)}) \neq P(\bx_{i}^{(k)})$, i.e., the marginal covariate distributions differ. Both parameter and residual shift can be seen as \emph{label shift}, meaning that the conditional label distributions $P(y_{i}^{(0)} | \bx_{i}^{(0)})$ and $P(y_{i}^{(k)} | \bx_{i}^{(k)})$ differ. Many examples have revealed that blindly transferring sources under distribution shift may pose challenges to the final performance and lead to the negative transfer \citep{wang2019characterizing, zhang2022survey}. Hence, the core issue of this paper is:
\begin{center}
    \emph{
        How to identify the transferable information from different source studies under distribution (covariate/parameter/residual) shift? How to conduct valid prediction and inference for high-dimensional quantile regression with knowledge transfer?
    }
\end{center}

Closely related to transfer learning, numerous approaches have been developed for multi-source data analysis, including federated learning \citep{mcmahan2017communication}, distributed learning \citep{jordan2018communication, fan2023communication}, fusion learning \citep{shen2020fusion, li2023efficient}, and multi-task learning \citep{duan2023adaptive, knight2024multi}. However, distinguished from multi-source settings where each dataset is equally important, transfer learning aims to leverage information from source domains to improve the performance of the target task. Statistical investigations into transfer learning have spanned diverse contexts, including linear regression \citep{li2022transfer, tian2022transfer}, nonparametric classification \citep{cai2021transfer}, nonparametric regression \citep{cai2024transfer}, graphical models \citep{li2023transfer}, and causal inference \citep{wei2023transfer}.

To address parameter shift, \citep{li2022transfer} introduced the contrast vector $\bdelta^{(k)}:=\bw^{(k)}-\sbeta$ and defined the transferable set $A_h := \{ \|\bdelta^{(k)}\|_1 \le h \}$, where source studies with a small contrast of parameters are considered informative. They proposed a two-step transfer procedure to first obtain an integrated estimate based on informative sources, then debias this estimate on the target data. The two-step framework is further extended to generalized linear regression \cite{tian2022transfer} and quantile regression \cite{zhang2022transfer,huang2023estimation}. In addition, several studies have investigated covariate shift and proposed remedies in high-dimensional linear regression \citep{li2023estimation, Zhao2026Residual, he2024transfusion}. 

To the best of our knowledge, none have explicitly accounted for residual shift. However, many applications highlight the necessity of addressing residual shift in heterogeneous settings, such as in image recognition \cite{park2023label} and disease diagnosis \cite{chen2022estimating}.
Here a motivating example is provided to illustrate the potential negative transfer caused by residual shift. We consider a similar setting to the simulation in \cite{zhang2022transfer} with no parameter or covariate shift. Specifically, we assume only one source study with different residual distributions. Figure \ref{fig:counter example} compares the average $\ell_2$-errors of four methods for estimating the target parameter at $0.2$-th quantile. More details can be found in the supplementary materials. 

\begin{figure}[ht!] 
    \begin{center}
        \includegraphics[width=1\linewidth]{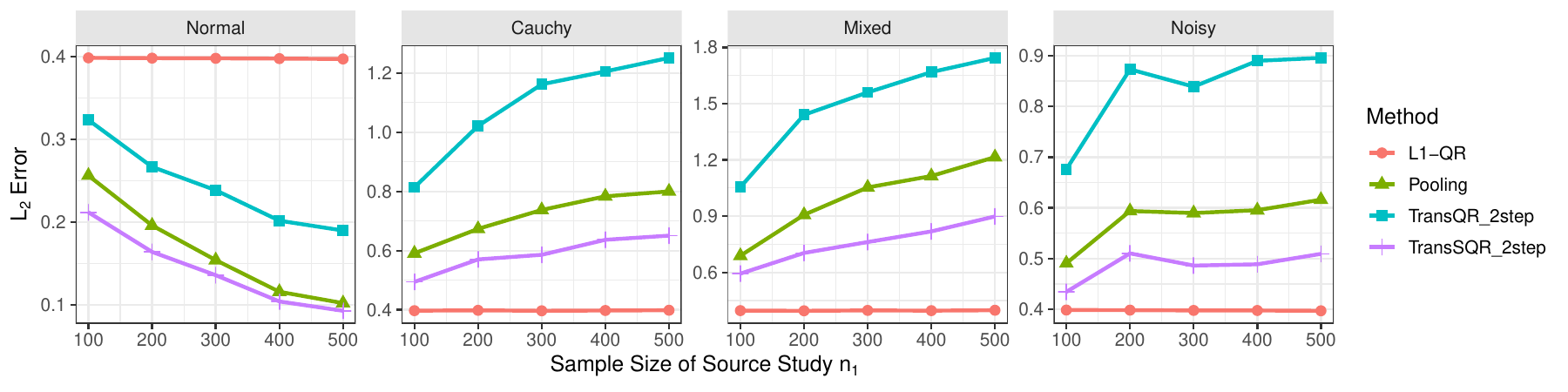}
    \end{center}
    \caption{Average $\ell_2$-errors at $0.2$-th quantile, where source residual distributions are: 1) Normal: $\mathcal{N}(0,1)$; 2) Cauchy: $\mathcal{C}(0,3)$; 3) Mixed: mixed $\mathcal{N}(-3,0.5)$ and $\mathcal{N}(3,0.5)$; 4) Noisy: $\mathcal{N}(0,5^2)$.}
    \label{fig:counter example}
\end{figure}

As shown in Figure \ref{fig:counter example}, apart from the normal setting that demonstrates the anticipated effect, all other settings reveal significant negative transfer, strongly suggesting that assessing transferability based solely on parameter similarity is insufficient in cases involving residual shift. Notably, residual shift fundamentally reflects the information quality of a source study around the objective quantile, where this difference can be characterized by the local density and sample size. Accordingly, we introduce a novel transferable set in Section \ref{section2} that simultaneously accounts for parameter and residual shift. This set is designed to filter out source studies with large parameter discrepancies or limited information at the $\tau$-th quantile. With this transferable set, we develop a comprehensive transfer learning framework to address parameter/covariate/residual shift concurrently for high-dimensional quantile regression. 

Our contributions are threefold. First, we introduce a new transferable set to handle both parameter and residual shift, offering deeper insights into managing distribution shift under knowledge transfer. For parameter estimation, we extend the constrained $\ell_1$-minimization \citep{li2023estimation} to quantile regression, ensuring robustness in the presence of covariate shift. We further develop a corresponding source detection procedure and integrate it into a unified transfer estimation algorithm. Second, in theory, we establish non-asymptotic upper error bounds for our transfer estimators, broadening existing results to account for residual shift in addition to parameter shift. We also derive the minimax lower error bound to demonstrate the rationale of our approach, and provide the detection consistency of the source screening procedure. Third, a debiased transfer approach based on Neyman orthogonality \citep{belloni2019valid} is proposed for statistical inference on quantile coefficients, which to our knowledge, is the first to debias on informative sources to enhance inference efficiency under distribution shift. By integrating valuable source studies, our approach achieves sharper asymptotic rates than existing results with $\sqrt{n_0}$-normality \citep{tian2022transfer, huang2023estimation}, leading to more efficient debiased estimates and shorter confidence intervals.

The rest of the article is organized as follows. In Section \ref{section2}, we introduce our transfer framework under distribution shift. Section \ref{section3} describes the source detection procedure. In Section \ref{section4}, we present the orthogonal debiased approach for statistical inference with knowledge transfer. Simulation studies and an application to Genotype-Tissue Expression (GTEx) dataset are provided in Section \ref{section5} and \ref{section6}, respectively.

\section{Transfer Learning under Distribution Shift} \label{section2}

We begin this section with some necessary notations. For two sequences of positive numbers $\left\{a_n\right\}$ and $\left\{b_n\right\}$, we write $a_n \lesssim b_n$ if $a_n \leq c b_n$ for some universal constant $c \in(0, \infty)$, and $a_n \gtrsim b_n$ if $a_n \geq c^{\prime} b_n$ for some universal constant $c^{\prime} \in(0, \infty)$. We say $a_n \asymp b_n$ if $a_n \lesssim b_n$ and $a_n \gtrsim b_n$. For any $q \in[0, \infty]$ and vector $\bx=(x_1, \cdots, x_p) \in \mathbb{R}^p$, we write $\|\bx\|_q$ for its $l_q$-norm. Let $a \wedge b$ and $a \vee b$ denote the minimum and maximum of $a$ and $b$, respectively. For an index subset $S \subseteq \{1, 2, \dots, p\}$, let $|S|$ denote its cardinality, and let $S^c$ denote its complement, i.e., $S^c = \{1, 2, \dots, p\} \setminus S$. For a vector $\bv \in \mathbb{R}^p$, we use $\bv_S$ to denote the restriction of vector $\bv$ to the index set $S$. Denote $\boldsymbol{e}_j=(\boldsymbol{0}_{j-1}^\top,1,\boldsymbol{0}_{p-j}^\top)^\top$ as the unit vector. The sub-Gaussian norm of a random variable $u \in \mathbb{R}$ is $\|u\|_{\psi_2}=\inf\{t>0:\mathbb{E}[\exp(u^2/t^2)] \le 2\}$ and the sub-Gaussian norm of a random vector $\boldsymbol{U} \in \mathbb{R}^n$ is $\|\boldsymbol{U}\|_{\psi_2}=\sup _{\|\bv\|_2=1, \boldsymbol{v} \in \mathbb{R}^n}\|\langle\boldsymbol{U}, \bv\rangle\|_{\psi_2}$. The expectation symbol $\mathbb{E}[(\cdot)_i^{(k)}]:=\mathbb{E}_{(\bx_{i}^{(k)}, y_i^{(k)})}[(\cdot)_i^{(k)}]$ means obtaining expectation with the joint distribution $P(\bx_{i}^{(k)}, y_i^{(k)})$.

\subsection{Methodology of Tackling Three Types of Distribution Shift} \label{section2-1} 

\textbf{Parameter Shift.} Naturally we should consider a source study as informative if its parameter contrast $\bdelta^{(k)}$ is relatively small. We choose the $\ell_1$-sparsity to quantify the informative level, as in many practical applications, parameter shift typically occurs across multiple dimensions without the overall magnitude  growing too fast. Formally, the screening criterion is to reject source studies with $\| \bdelta^{(k)} \|_1$ larger than an oracle transferring level, which has been widely adopted in the previous transfer work \citep{li2022transfer, tian2022transfer}.  

\noindent \textbf{Residual Shift.} We start from the simplest low-dimensional linear regression to understand how to tackle the residual shift. Suppose the target and source model as
\begin{align*}
	\left\{
	\begin{array}{l}
		\text{target model: } ~ y_i^{(0)}=(\bx_i^{(0)})^\top\sbeta+\epsilon_i^{(0)}, ~ i=1,\dots, n_0, ~ \epsilon_i^{(0)} \sim \mathcal{N}(0, \sigma_0^2), \\
		\text{source model: } ~ y_i^{(1)}=(\bx_i^{(1)})^\top\sbeta+\epsilon_i^{(1)}, ~ i=1,\dots, n_1, ~ \epsilon_i^{(1)} \sim \mathcal{N}(0, \sigma_1^2).
	\end{array}
	\right.
\end{align*}
Denote the covariance matrix $\Sigma_x:=\mathbb{E}[\bx_i^{(0)}(\bx_i^{(0)})^\top]=\mathbb{E}[\bx_i^{(1)}(\bx_i^{(1)})^\top]$. Based on the knowledge of linear models, the target and pooling OLS estimator $\wbeta_{tar}, \wbeta_{pool}$ satisfy:
\small
\begin{align*}
	\left(\wbeta_{tar} - \sbeta\right) \sim \mathcal{N}\left(0, \frac{\sigma_0^2}{n_0}\Sigma_x^{-1}\right), \quad \left(\wbeta_{pool} - \sbeta \right) \sim \mathcal{N}\left(0, \frac{n_0 \sigma_0^2 + n_1 \sigma_1^2}{(n_0 +  n_1)^2} \Sigma_x^{-1}\right).
\end{align*}
\normalsize
If $\sigma_1^2$ is much larger than $\sigma_0^2$, i.e., the source study is much more noisy, then the pooling estimator will be much less efficient compared to the target one. 

Moving to the quantile regression with traditional loss $\rho_{\tau}(x):=x(\tau-\mathbb{I}(x\le0))$, assume $\epsilon_i^{(0)} \perp \bx_i^{(0)}, \epsilon_i^{(1)} \perp \bx_i^{(1)}$ and $\mathbb{P}(\epsilon_i^{(0)} \le 0) = \mathbb{P}(\epsilon_i^{(1)} \le 0) = \tau$, denote the density functions of $\epsilon_i^{(0)}, \epsilon_i^{(1)}$ as $f_0(\cdot), f_1(\cdot)$, respectively. Under some regular conditions, it is easy to verify that
\small
\begin{align*} 
    \left( \wbeta_{tar}-\sbeta \right) \xrightarrow{d} \mathcal{N}\left(0, \frac{\tau(1-\tau)}{n_0 (f^{(0)}(0))^2} \Sigma_x^{-1}\right), ~
    \left( \wbeta_{pool} - \sbeta \right) \xrightarrow{d} \mathcal{N}\left(0, \frac{\tau(1-\tau) (n_0 + n_1)}{(n_0 f^{(0)}(0) + n_1 f^{(1)}(0))^2} \Sigma_x^{-1}\right).
\end{align*}
\normalsize
If $f^{(1)}(0)$ is much smaller than $f^{(0)}(0)$, for example, when $n_1 f^{(1)}(0) \le (\sqrt{1 + n_1 /n_0} - 1) n_0 f^{(0)}(0)$, $\wbeta_{pool}$ will be less efficient than $\wbeta_{tar}$. The residual density values $f^{(0)}(0)$ and $f^{(1)}(0)$ serve a similar role for QR estimators as residual variances $\sigma_0^2$ and $\sigma_1^2$ do for OLS estimators, both indicating the quality of information available for estimation. Also, information quality depends on the sample size, since given the density function, a larger sample size provides richer quantile information. Consequently, another screening criterion should exclude source studies with small sample sizes $n_k$ and low densities at the $\tau$-th quantile $f^{(k)}(0 | \bx_{i}^{(k)})$, as they are unlikely to contribute valuable information beyond introducing noise. For simplicity, we denote $f^{(k)}_i := f^{(k)}(0 | \bx_{i}^{(k)})$ as the conditional density of residual $\epsilon_{i}^{(k)}$ at zero, hereinafter referred to as the \emph{$\tau$-th density}.

\begin{remark}
    Rejecting sources with lower $\tau$-th density $f^{(k)}_i$, i.e., leveraging sources with higher $\tau$-th density, can also be explained as reducing external variance. If a source study has a negligible variance term, i.e., $\sigma_1^2=0$, then it means accurately identifying the target parameter, leading to a wonderful transferable study. From the perspective of minimizing the mean square error (MSE), the transfer estimator benefits from a smaller contrast $\| \bdelta^{(k)} \|_1$ to reduce bias, while a larger $\tau$-th density helps to control variance. We further extend this motivating case to the covariate shift setting, where similar insights hold; please see Section 8.1.1 of supplementary materials for details.
\end{remark} 

Building on the previous analysis, we propose the oracle transferable set $\bc_h$ as:
\begin{align} \label{oracle trans set}
	\bc_h= \left\{ k: \quad 1 \le k \le K,  \quad  \| \bdelta^{(k)} \|_1 \le h_1, \quad \frac{n_0 \mathbb{E}[f^{(0)}_i]}{n_k \mathbb{E}[f^{(k)}_i]} \le h_2 \right\},
\end{align}
where $h_1$ and $h_2$ are oracle transferability levels characterizing the magnitude of parameter and residual shift, respectively. For each source study, the term $\| \bdelta^{(k)} \|_1$ reflects parameter similarity, while the term $n_k \mathbb{E}[f^{(k)}_i]$ captures the information quality around the objective quantile. The multiplicative form of $n_k \mathbb{E}[f^{(k)}_i]$ is motivated by the structure $\sum_{k \in {0} \cup \bc_h} n_k \mathbb{E}[f_i^{(k)}]$ appearing in Theorem \ref{thm:convergence} in the next subsection. To guarantee the improvement of transfer learning, we only leverage the information from source studies included in this set $\bc_h$. It is important to note that, unlike $h_1$, which is expected to converge to zero as in the previous literature \citep{li2022transfer, tian2022transfer}, $h_2$ depends on the sample sizes of the target and source studies, thus is not necessary to converge to zero or diverges to infinity.

\noindent \textbf{Covariate Shift.}  The negative effect of covariate shift is mainly brought by the divergence of covariance matrices $\bSigma^{(k)}:=\mathbb{E}[\bx_{i}^{(k)}(\bx_{i}^{(k)})^\top]$ between target and source studies. As mentioned by \cite{li2022transfer, he2024transfusion}, in the two-step transfer approach with transferable set $\bc_h$, the pooling estimator in the first step makes the contrast $\bdelta^{(k)}$ amplified by the factor $C_{\bSigma}:=1+|\bc_h|\max_{k \in \bc_h}\sup_{1 \le j \le p}\|\boldsymbol{e}_j^\top (\bSigma^{(k)}-\bSigma^{(0)})(\sum_{k \in \bc_h}\bSigma^{(k)})^{-1}\|_1$, which may diverge when source covariance matrices $\bSigma^{(k)}$ are dissimilar to the target $\bSigma^{(0)}$. To address this problem, here we generalize the constrained optimization algorithm proposed by \cite{li2023estimation} to the field of quantile regression, for jointly estimating the target parameter $\sbeta$ and contrast vectors $\{\bdelta^{(k)}\}_{k \in \bc_h}$. Our transfer framework is formalized as:
\small
\begin{equation} \label{l1 transfer}
    \begin{aligned}
        & \left(\wbeta, \{\wdelta^{(k)}\}_{k\in \bc_h}\right) = \underset{\bbeta, \{\bdelta^{(k)}\}_{k\in \bc_h}}{\arg \min} \lambda_{\bbeta} \left\| \bbeta \right\|_1 + \sum_{k\in \bc_h} \lambda_k \| \bdelta^{(k)} \|_1, \\ 
        & \text{subject to } 
        \left\{ \begin{array}{l}
            \| S_{n_k} (\bbeta + \bdelta^{(k)}) \|_\infty \le \lambda_k, \quad \text{for } k \in  \{0\} \cup \bc_h, \\
            \| \sum_{k \in \{0\} \cup \bc_h} S_{n_k} (\bbeta + \bdelta^{(k)}) \|_\infty \le \lambda_{\bbeta},
        \end{array} \right. 
    \end{aligned}
\end{equation}
\normalsize
where $S_{n_k}(\bbeta):=\partial L^{(k)}(\bbeta) / \partial \bbeta$ is the subgradient of the quantile loss function $L^{(k)}(\bbeta):=\sum_{i=1}^{n_k}\rho_{\tau}(y_i^{(k)}-(\bx_i^{(k)})^{\top}\bbeta)$, and $\lambda_{\bbeta}, \{ \lambda_k \}_{k \in\{0\} \cup \bc_h}$ are tuning parameters. To solve this constrained $\ell_1$-minimization problem \eqref{l1 transfer}, we can follow the iterative algorithm adopted in \cite{li2023estimation}, where the detailed procedure is summarized in Section 8.1.2 of supplementary materials. Notably, the $\ell_1$-sparsity regularizer for each contrast $\bdelta^{(k)}$ allows us to be free from the similarity restriction of covariance matrices, hence removing the $C_{\bSigma}$ factor and achieving robustness under covariance heterogeneity \citep{li2023estimation, he2024transfusion}. 

\begin{remark}
    In our transfer framework \eqref{l1 transfer}, the objective function encourages sparse solutions of $\wbeta$ and $\{\wdelta^{(k)}\}_{k\in \bc_h}$. The constraint $\| S_{n_0} (\bbeta) \|_{\infty} \le \lambda_0$ is inherited from the target model, imposing that $\bbeta$ should be identified as the true parameter $\sbeta$ in the target task. The constraint $\| S_{n_k} (\bbeta + \bdelta^{(k)}) \|_{\infty} \le \lambda_k$ comes from the $k$-th source model, imposing that $\bdelta^{(k)}$ should be identified as the true contrast $\bw^{(k)}-\sbeta$. The last constraint $\| \sum_{k \in \{0\} \cup \bc_h} S_{n_k} (\bbeta + \bdelta^{(k)}) \|_\infty \le \lambda_{\bbeta}$ aggregates target study and studies in the oracle transferable set $\bc_h$, revealing that our estimation leverages information from valuable sources. In contrast to the two-step transfer framework \citep{li2022transfer, tian2022transfer}, our algorithm can simultaneously estimate $\sbeta$ and $\bdelta^{(k)}$, providing more information for the downstream inference task.
\end{remark}

\subsection{Theory on Estimation Error}

Now we study the non-asymptotic convergence rate of $\wbeta - \sbeta$ obtained by the transfer approach \eqref{l1 transfer} with sources in the index set $\bc$, where $\bc$ is an arbitrary prespecified set for theoretical generality. Under the high-dimensional sparse setting, we characterize the geometric structure of $\wbeta - \sbeta$ and $\{\wdelta^{(k)}-\bdelta^{(k)}\}_{k=1}^K$ through restrictive sets \citep{belloni2011l1, wang2024analysis}, denoted by
\small
\begin{align*}
	\text{for $\wbeta - \sbeta$ with $\ell_0$-sparsity: \quad} & \Gamma_H =\left\{v \in \mathbb{R}^p:\left\|v_{S^c}\right\|_1 \leq \left\|v_S\right\|_1\right\}, \\
	\text{for $\wdelta^{(k)}-\bdelta^{(k)}$ with $\ell_1$-sparsity: \quad} & \Gamma_W^{(k)} = \{ v \in \mathbb{R}^p: \| v_{S_k^c} \|_1 \leq \| v_{S_k} \|_1 + 2 \| \bdelta^{(k)}_{S_k^c} \|_1 \},
\end{align*}
\normalsize
where $S=\{j: \bbeta^*_j \neq 0, 1 \leq j \leq p\}$, $S_k = \{ j: |\bdelta^{(k)}_j| > a_k, 1 \leq j \leq p\}$ with a threshold $a_k > 0$. As noted in \citep{wang2024analysis}, $a_k$ controls the approximate sparsity level of the contrast vector $\bdelta^{(k)}$. Furthermore, for each contrast vector $\bdelta^{(k)}$, we denote $s_k = \| \bdelta^{(k)}_{S_k} \|_0, \eta_k = \| \bdelta^{(k)}_{S_k^c} \|_1$. For simplicity, we omit the dependence on $a_k$ in the notations $\Gamma_W^{(k)}, s_k, \eta_k$. We introduce three conditions below, where the index $k$ refers to each study $k \in \{0\} \cup \bc$.

\begin{condition} \label{cond:restricted eigenvalue}
    There exist universal positive constants $m_1, m_2$ such that
    \begin{align*}
        \lambda_{\min}(\mathbb{E}[f_i^{(k)}\bx_{i}^{(k)}(\bx_{i}^{(k)})]) \ge  m_1 \mathbb{E}[f_i^{(k)}], \quad \lambda_{\max}(\mathbb{E}[f_i^{(k)}\bx_{i}^{(k)}(\bx_{i}^{(k)})]) \le m_2 \mathbb{E}[f_i^{(k)}],
    \end{align*}
    where $\lambda_{\min}(\cdot), \lambda_{\max}(\cdot)$ denote the minimum and maximum matrix eigenvalues, respectively.
\end{condition}

\begin{condition} \label{cond:sub-gaussian design}
    Each covariate $\bx_i^{(k)}$ is a mean-zero $m_g$-sub-Gaussian random vector, where the sub-Gaussian norm $m_g$ is a universal positive constant. 
\end{condition}

\begin{condition} \label{cond:density} 
    The conditional density function $f^{(k)}(\cdot|\bx_{i}^{(k)})$ is positive, continuously differentiable, and uniformly upper bounded. Also, there exist a universal positive constant $b_0$ such that the derivative $f_i^{(k)'}(t|\bx_{i}^{(k)})$ is uniformly bounded for $|t| \le b_0$.
\end{condition}

Condition \ref{cond:restricted eigenvalue} assumes a modified version of positive definite covariance matrices with incorporating the $\tau$-th density $\mathbb{E}[f_i^{(k)}]$. The reason for placing $\mathbb{E}[f_i^{(k)}]$ on the right-hand side is to technically derive error bounds involving each study's $\tau$-th density. Note that if the study model is homoscedastic, then this condition will be the same as Condition 3.1 in \cite{li2023estimation}. Condition \ref{cond:sub-gaussian design} needs covariates following sub-Gaussian distributions, which relaxes the former requirement of bounded design \citep{tian2022transfer, huang2023estimation}. Condition \ref{cond:density} requires the boundness and smoothness of conditional density functions. Here a uniform lower bound of $f_i^{(k)}$ is not assumed, since we want to reveal the effect of residual shift in the error bounds. All three conditions are very common in the high-dimensional QR literature \citep{belloni2011l1, wang2024analysis}.

\begin{theorem} \label{thm:convergence}
	Assume Conditions \ref{cond:restricted eigenvalue}-\ref{cond:density} are satisfied. Let $\wbeta$ be obtained from the transfer framework \eqref{l1 transfer} with sources in the index set $\bc$. Define $\nc=n_0 + \sum_{k \in \mathcal{C}}n_k$. Let $\lambda_0 = c_\lambda n_0 \sqrt{\log p / n_0}$, $\lambda_k = c_\lambda n_k \sqrt{\log p / (n_k \wedge n_0)}$, $a_k = \lambda_k / (n_k \mathbb{E}[f_i^{(k)}])$ for each source study $k \in \bc$, and $\lambda_{\bbeta} = c_\lambda \sqrt{\nc \log p} + \sqrt{\sum_{k \in \{0\} \cup \mathcal{C}} n_k \mathbb{E}[f_i^{(k)}]} \sqrt{\sum_{k \in \mathcal{C}} \lambda_k \| \bdelta^{(k)} \|_1 / s}$,
    where $c_\lambda$ is a large enough constant. Also denote $\bar{s} = \sup_{k \in \bc} s_k, \underline{f} = \inf_{k \in \bc} \mathbb{E}[f_i^{(k)}], \bar{n} = \sup_{k \in \{0\} \cup \bc} n_k$, and 
    \small
    \begin{align*}
        r_n = s + \left( \sum_{k\in\bc} \frac{\lambda_k}{\lambda_{\bbeta}} \right)^2 \bar{s} + \left( \lambda_{\bbeta} \sqrt{\frac{\log p}{n_0}} \right)^{-1} \left( \frac{\bar{n}|\bc|}{n_0\underline{f}} \bar{s} \log p + \sum_{k\in\bc}\lambda_k\eta_k \right). 
    \end{align*}
    \normalsize
    If $r_n^2 \log p / \sqrt{n_0} = O(1)$, then for any given $\varepsilon>0$, with probability at least $1- \varepsilon - |\bc| A_1 \exp( - A_2 \log p)$, we have that
    \small
	\begin{align}
		\label{l2 error}
		\| \wbeta - \sbeta \|_2 & \le  C_1 \frac{\sqrt{n_\bc s \log p}}{\sum_{k \in \{0\} \cup \mathcal{C}} n_k \mathbb{E}[f_i^{(k)}]}   + C_2 \sqrt{ \frac{\sum_{k \in \mathcal{C}} n_k \sqrt{\frac{\log p}{n_k \wedge n_0}} \| \bdelta^{(k)} \|_1} {\sum_{k \in \{0\} \cup \mathcal{C}} n_k \mathbb{E}[f_i^{(k)}] }}, \\
		\label{l1 error}
		\| \wbeta - \sbeta \|_1 & \le C_1 \frac{s \sqrt{n_\bc \log p}}{\sum_{k \in \{0\} \cup \mathcal{C}} n_k \mathbb{E}[f_i^{(k)}]}   + C_2 \sqrt{ \frac{\sum_{k \in \mathcal{C}} s n_k \sqrt{\frac{\log p}{n_k \wedge n_0}}\| \bdelta^{(k)} \|_1} {\sum_{k \in \{0\} \cup \mathcal{C}} n_k \mathbb{E}[f_i^{(k)}] }},
	\end{align}
    \normalsize
	where $A_1, A_2, C_1, C_2$ are positive constants.
\end{theorem}

\begin{remark}
    The assumption $r_n^2 \log p / \sqrt{n_0} = O(1)$ is a technical requirement for parameter growth, where $r_n \sqrt{\log p /n_0}$ is a preliminary $\ell_1$-error bound for $\wbeta - \sbeta$ in the proof. A simple sufficient condition implied by this growth requirement is
    \begin{align*}
        \frac{s^2}{\underline{f}^2} \log p = O(n_0^{1/3}), \quad
        \frac{\bar{s}^2}{\underline{f}^2} \log p = O(n_0^{1/3}), \quad
        \frac{(\bar{n} |\bc|)^2}{n_0 n_{\bc}} = O(n_0^{1/12}), \quad
        \frac{\sup_{k \in \bc} \eta_k }{\sqrt{n_{\bc}}} = O(n_0^{-3/4}).
    \end{align*}
    This condition is not restrictive in typical high-dimensional regimes, especially when the study sizes are comparable and the contrast vectors $\{\bdelta^{(k)}\}_{k \in \bc}$ are sparse. Our analysis allows the size of source set $|\bc|$ goes to infinity.
\end{remark}

The error bounds of $\wbeta-\sbeta$ in \eqref{l2 error}/\eqref{l1 error} contain two parts. The first part can be seen as the inherent estimation error $\sqrt{s \log p / n_{\bc}}$ weighted by the variance term $\sum_{k \in \{0\} \cup \mathcal{C}} n_k \mathbb{E}[f_i^{(k)}]$ from residual shift. The second part represents the bias term involving each contrast $\| \bdelta^{(k)}\|_1$ from parameter shift weighted by the variance term. The error bounds reveal that to improve the transfer performance, source studies need to possess small $\| \bdelta^{(k)}\|_1$ as well as large $n_k \mathbb{E}[f_i^{(k)}]$, strongly confirming the rationality of our screening criteria \eqref{oracle trans set}. Here we give some quantitative explanations under simplified settings, taking the $\ell_2$-error bound \eqref{l2 error} as the example. 

\begin{itemize}[leftmargin=*, itemsep=0pt]
    \item (\textbf{Parameter Shift}) If the $\tau$-th densities $\mathbb{E}[f_i^{(k)}]$ are neglected and $n_0 = O(n_k), k \in \bc$, then the error bound can be simplified to 
    $\sqrt{s \log p / n_\bc} + (\log p/n_0)^{1/4} \sqrt{\max_{k \in \bc} \| \bdelta^{(k)} \|_1}$. 
    Parallel to the previous transfer theory, the first term $\sqrt{s \log p / n_\bc}$ enjoys a sharper rate than the single-task $\ell_2$-error $\sqrt{s \log p /n_0}$ \citep{belloni2011l1, wang2024analysis}, which indeed reflects the superiority of transfer learning. The second term is a bias term by parameter shift, corresponding to the result in Theorem 1 of \cite{li2023estimation}. 
    \item (\textbf{Residual Shift}) To our knowledge, this non-asymptotic result containing $\tau$-th density $\mathbb{E}[f_i^{(k)}]$ is new to the transfer literature, since we keep the density items along our proving process. For both terms, we can observe clearly from the denominator that the contribution of each study is represented by the product $n_k \mathbb{E}[f_i^{(k)}]$, which reflects the information quality around the objective quantile and reinforces the reasonability of our transferability level $h_2$ defined in \eqref{oracle trans set}. For example, consider the case with only one source study that exhibits substantial residual shift such that $n_1 \mathbb{E}[f_i^{(1)}] < (\sqrt{1 + n_1 /n_0} - 1) n_0 \mathbb{E}[f_i^{(0)}]$, then incorporating this source study will lead to negative transfer.
    \item (\textbf{Covariate Shift}) Our framework can relax the elusive conditions about the covariance matrices across studies (Assumption 4 in \cite{tian2022transfer}, Condition 4 in \cite{zhang2022transfer}, Assumption 4 in \cite{huang2023estimation}). Specifically, our results neither assume the similarity of $\bSigma^{(k)}$, nor involve the factor $C_{\bSigma}$ mentioned above, hence achieving better performance under covariate shift.
\end{itemize}

Applying Theorem \ref{thm:convergence} to the transferable set $\bc_h$, together with the definitions of $h_1$ and $h_2$ in \eqref{oracle trans set}, we can immediately obtain the following corollary.

\begin{corollary} \label{cor:oracle converence rate}
    Under Conditions \ref{cond:restricted eigenvalue}--\ref{cond:density} and the assumptions of Theorem \ref{thm:convergence}, suppose further that $|\bc_h| n_0 \lesssim n_{\bc_h}$, $h_2 = O(|\bc_h|)$, and let $\lambda_{\bbeta} = c_\lambda \sqrt{n_{\bc_h} \log p} + \sqrt{n_{\bc_h} h_1} \left(\frac{\log p}{n_0}\right)^{1/4} \sqrt{\frac{|\bc_h| n_0 \mathbb{E}[f_i^{(0)}]}{h_2 s}}$. Denote by $\wbeta_{ora}$ the estimator obtained from \eqref{l1 transfer} with the oracle transferable set $\bc_h$. Then for any constant $\varepsilon>0$,
    \small
    \begin{align*}
        \| \wbeta_{ora} - \sbeta \|_2 \le C_1 \sqrt{\frac{s \log p}{n_{\bc_h}}} \frac{n_{\bc_h} h_2}{|\bc_h| n_0 \mathbb{E}[f_i^{(0)}]} + C_2 \sqrt{h_1} \left( \frac{\log p}{n_0} \right)^{\frac{1}{4}}  \sqrt{\frac{n_{\bc_h} h_2}{|\bc_h| n_0 \mathbb{E}[f_i^{(0)}]}}.
    \end{align*}
    \normalsize
    with probability at least $1- \varepsilon - |\bc_h| A_1 \exp( - A_2 \log p)$.
\end{corollary}

As shown in Corollary \ref{cor:oracle converence rate}, the quantities $h_1$ and $h_2$ govern the levels of parameter and residual shift, respectively, and jointly determine the trade-off between transfer gains and the risk of negative transfer. A sufficient condition for avoiding negative transfer, namely, for achieving an error rate sharper than $\sqrt{n_0 s \log p} / (n_0 \mathbb{E}[f_i^{(0)}])$, is 
\begin{align*}
    h_1 = o\left(s \frac{\sqrt{n_0 \log p}}{n_0 \mathbb{E}[f_i^{(0)}]}\right), \qquad  h_2 = O\left(\frac{|\bc_h| n_0}{n_{\bc_h}}\right). 
\end{align*}
This condition implies that the parameter contrast $\|\bdelta^{(k)}\|_1$ for source studies in the oracle transferable set should be smaller than the $\ell_1$-error rate of $\ell_1$-QR on the target study, while the information quality $n_k \mathbb{E}[f_i^{(k)}]$ should not be smaller than the order of $n_{\bc_h} \mathbb{E}[f_i^{(0)}] / |\bc_h|$, which can be interpreted as the average study size multiplied by the target-study $\tau$-th density. Compared with \cite{li2023estimation}, this result generalizes the existing theory by allowing error control through both parameter and residual shift. To further illustrate the role of the $\tau$-th density in the upper bound, we next establish a minimax lower bound for the $\ell_2$-estimation error.



\begin{theorem} \label{thm:minimax} 
    Assume Conditions \ref{cond:restricted eigenvalue}-\ref{cond:density} hold. Suppose the parameter space
    \small
    \begin{align*}
        \bTheta (s, h_1, h_2) = \left\{ \boldsymbol B = 
        \left(\bbeta, \{\bdelta^{(k)}\}_{k=1}^K\right) \text{ satisfying \eqref{qr}: }   
        \| \bbeta \|_0 \le s, \| \bdelta^{(k)} \|_1 \le h_1,
        \frac{n_0 \mathbb{E}[f_i^{(0)}]}{n_k \mathbb{E}[f_i^{(k)}]} \le h_2, 
        k \in \bc_h
        \right\}.
    \end{align*}
    \normalsize
    If $h_2 = O(|\bc_h|)$, $\log(p/s)\asymp\log p$ and $ (s \log p) / (n_0 \mathbb{E}[f_i^{(0)}]) = o(1)$, then
    \small
    \begin{align*}
        & \inf_{\wbeta}
        \sup_{\boldsymbol B \in \bTheta(s,h_1,h_2)}
        \mathbb P
        \left[
        \|\wbeta-\bbeta\|_2
        \gtrsim \right. \\
        & \qquad \left. \sqrt{\frac{s \log p}{n_{\bc_h}}} \cdot \sqrt{\frac{n_{\bc_h} h_2}{|\bc_h| n_0 \mathbb{E}[f_i^{(0)}]}} + \sqrt{\frac{s \log p}{n_0}} \wedge \left( h_1^2 \frac{\log p}{n_0} \right)^{1/4} \wedge h_1
        \right]
        \ge \frac{1}{2},
    \end{align*}
    \normalsize
    where $\wbeta$ is an arbitrary estimator obtained from all the data.
\end{theorem}

To interpret the minimax lower bound in Theorem \ref{thm:minimax}, we emphasize three points. 
First, if the residual shift component is ignored, or equivalently if the $\tau$-th conditional densities of all transferable source studies are uniformly bounded away from zero, the lower bound reduces to $\sqrt{s \log p / n_{\bc_h}} + \sqrt{s \log p / n_0} \wedge (h_1^2 \log p / n_0)^{1/4} \wedge h_1$. This is the same rate as that established for high-dimensional linear transfer learning \citep{li2022transfer,tian2022transfer,li2023estimation}, recovering the classical transfer learning minimax result. Moreover, when $\sqrt{\log p / n_0} \lesssim h_1 \lesssim s \sqrt{\log p / n_0}$, the lower bound matches the upper bound in Corollary \ref{cor:oracle converence rate}.
Second, once residual shift is taken into account, the first term of the lower bound explicitly involves the information quality parameter $h_2$. This shows that, even when the parameter shift within the transferable set is well controlled, transfer efficiency still depends on the local information carried by the source studies around the target quantile, especially when the source studies have moderate sample sizes. 
Third, we note that the lower bound differs from the upper bound in Corollary \ref{cor:oracle converence rate} by a factor of order $\sqrt{\frac{n_{\bc_h}h_2}{|\bc_h|n_0\mathbb E[f_i^{(0)}]}}$. This gap arises because the upper and lower bounds are derived from different perspectives. The upper bound uses the curvature of the quantile loss, while the lower bound is based on the local Fisher information, which captures the role of the quantile density differently. Although this gap indicates a remaining limitation of the current theory for high-dimensional quantile transfer learning under residual shift, the lower bound is already sufficient to demonstrate that controlling the information quality of source studies is necessary for beneficial transfer.

\section{Transferable Set Detection} \label{section3}

Since the oracle transferable set $\bc_h$ may be unknown in practice, here we propose a detection procedure to screen out the source studies satisfying the requirements in \eqref{oracle trans set}. We can first get an initial estimate $\tbeta$ through $\ell_1$-QR on the target study, then remove the part accounted by $\tbeta$ on each source study and conduct $\ell_1$-QR to obtain the contrast estimators $\{\tdelta^{(k)}\}_{k=1}^{K}$. With contrast estimators  $\{\tdelta^{(k)}\}_{k=1}^{K}$, screening studies with small $\| \bdelta^{(k)} \|_1$ is relatively simple. We can directly calculate each $\| \tdelta^{(k)} \|_1$ and set a threshold to reject the dissimilar studies. As implied by Corollary \ref{cor:oracle converence rate}, $h_1$ should be no more than the single-task error \citep{zhang2022transfer, huang2023estimation}, here we set the threshold as $t_1\sqrt{\log p / n_0}$, where $t_1$ is a thresholding parameter. The first screening set is then defined as $\wc_1=\{k: \| \tdelta^{(k)}\|_1 \le t_1\sqrt{\log p / n_0}\}$. We will show that under some regularity conditions, appropriate choice of $t_1$ will lead to consistent selection conforming with the requirement $\| \bdelta^{(k)} \|_1 \le h_1$. 

The second requirement of the oracle transferable set $\bc_h$ is to reject source studies with tiny information $n_k \mathbb{E}[f_i^{(k)}]$ compared to the target $n_0 \mathbb{E}[f_i^{(0)}]$ around the $\tau$-th quantile. We employ a kernel-based approximation using the Powell bandwidth $b_k$ \citep{dai2021inference, huang2023estimation, koenker2005quantile} to estimate each $\mathbb{E}[f_i^{(k)}]$, with the choice of $b_k$ detailed in the supplementary materials. Specifically, for each study we calculate 
\begin{align} \label{density approximation}
	\tf^{(k)} (\tbeta + \tdelta^{(k)}) = \frac{1}{2 b_k n_k} \sum_{i=1}^{n_k} \mathbb{I}\left( | y_i^{(k)}-(\bx_i^{(k)})^\top(\tbeta+\tdelta^{(k)}) | \le b_k \right).
\end{align} 
We then define the second-stage screening set as $\wc_{2} = \{k: n_k \tf^{(k)} (\tbeta + \tdelta^{(k)}) \ge t_2 n_0 \tf^{(0)} (\tbeta) \}$ with the thresholding parameter $t_2$. The threshold $t_2$ ensures that a source study is considered informative as long as its $\tau$-th density is comparable to the target, or its sample size is sufficiently large. We will show that under some regularity conditions, appropriate choice of $t_2$ will lead to consistent selection conforming with the requirement $(n_0 \mathbb{E}[f^{(0)}_i]) / (n_k \mathbb{E}[f^{(k)}_i]) \le h_2$. The final detected transferable set is $\wc = \wc_1 \cap \wc_2$. 

Note that the calculation of estimators $\tbeta, \{\tdelta^{(k)}\}_{k=1}^{K}$ is exactly the first round of iterative procedure for solving the $\ell_1$-minimization transfer framework \eqref{l1 transfer}. Hence here, we propose a unified one-shot algorithm for transferable set detection and parameter prediction, which is summarized in Algorithm \ref{alm:unify}. Note that for complex practical scenarios, such as strong distribution shift or heavy-tailed distributions, the procedure need to be iterated multiple times until the detected set has stabilized.

\begin{algorithm}[ht!] 
\small
\caption{Parameter prediction with transferable set detection}
\label{alm:unify}
\KwIn{Target data $\{((\bx_{i}^{(0)})^{\top}, y_{i}^{(0)})_{i=1}^{n_0}\}$, source data $\{((\bx_{i}^{(k)})^{\top}, y_{i}^{(k)})_{i=1}^{n_k}\}_{k=1}^{K}$, thresholds $t_1$ and $t_2$, tuning parameters $\lambda_{\bbeta}$ and $\{ \lambda_k \}_{0 \le k \le K}$, bandwidths $\{ b_k \}_{0 \le k \le K}$.}
\begin{enumerate}[itemsep=0pt]
    \item Compute $\tbeta \leftarrow \ell_1 \text{-QR } \text{ on } [((\bx_{i}^{(0)})^{\top}, y_{i}^{(0)})_{i=1}^{n_0},\lambda_0]$.
    \item Compute $\tdelta^{(k)} \leftarrow \ell_1 \text{-QR } \text{ on } [ ((\bx_{i}^{(k)})^{\top}, y_{i}^{(k)} - (\bx_{i}^{(k)})^{\top} \tbeta)_{i=1}^{n_k}, \lambda_k ], 1 \le k \le K$.
    \item Compute $\{\tf^{(k)} (\tbeta + \tdelta^{(k)})\}_{0 \le k \le K}$ with bandwidths $\{ b_k \}_{0 \le k \le K}$ via \eqref{density approximation}.
    \item Let $\wc_1=\{k: \| \tdelta^{(k)}\|_1 \le t_1\sqrt{\log p / n_0}\}, \wc_{2} = \{k: n_k \tf^{(k)} (\tbeta + \tdelta^{(k)}) \ge t_2 n_0 \tf^{(0)} (\tbeta ) \}$. 
    \item Determine the transferable set $\wc = \wc_1 \cap \wc_2$, then solve the pooled $\ell_1$-QR
    $$
    \wbeta = \underset{\bbeta}{\arg \min}
    \left\{\sum_{k \in \{0\} \cup \wc} L^{(k)} (\bbeta + \tdelta^{(k)}) + \lambda_{\bbeta} \left\| \bbeta \right\|_1 \right\}.
    $$
\end{enumerate}
\KwOut{$\wbeta$.}
\end{algorithm}

We next establish the detection consistency of $\wc$ towards the oracle but unknown transferable set $\bc_h$. To ensure the identifiability of source studies satisfying the oracle requirement in \eqref{oracle trans set}, we need to assume a significant gap of quantile coefficients and $\tau$-th densities, where similar conditions are also assumed in \cite{tian2022transfer, zhang2022transfer, huang2023estimation}.  

\begin{condition} \label{cond:delta gap}
	For each source study $k$, let $\Omega_1^{(k)} := C_3 \left(\frac{s_k}{\mathbb{E}[f_i^{(k)}]} \sqrt{\frac{ \log p}{n_k \wedge n_0}} + \eta_k\right)$, where $C_3$ is a positive constant specified in the supplementary materials. Denote the set $\bc_{h_1}:=\{k \in \{1,\dots,K\}: \| \bdelta^{(k)} \|_1 \le h_1\}$. $\bc_{h_1}$ satisfies that
    \small
	\begin{align*}
		\inf_{k \in \bc_{h_1}^c} \left( \| \bdelta^{(k)} \|_1 - \Omega_1^{(k)} \right)  > \sup_{k \in \bc_{h_1}} \left( \| \bdelta^{(k)} \|_1 + \Omega_1^{(k)} \right) . 
	\end{align*}
    \normalsize
\end{condition} 

\begin{condition} \label{cond:density gap}
	For each source study $k$, let 
	\small
    \begin{align*}
		\Omega_2^{(k)} := C_4 \left( \frac{1}{\sqrt{b_k n_k}}\sqrt{s + \frac{s_k}{(\mathbb{E}[f_i^{(k)}])^2} + \sqrt{\frac{s_k n_k}{\log p}} \frac{\eta_k}{\mathbb{E}[f_i^{(k)}]}} + \frac{1}{b_k} \left( s \sqrt{\frac{\log p}{n_0}} + \Omega_1^{(k)} \right) \right),
	\end{align*}
    \normalsize
	where $C_4$ is a positive constant specified in the supplementary materials. Also define $\Omega_2^{(0)} := C_4\left( \sqrt{\frac{s}{b_0 n_0}} + \frac{s}{b_0}\sqrt{\frac{\log p}{n_0}} \right)$. Denote the set $\bc_{h_2}:=\{k \in \{1,\dots,K\}: (n_0 \mathbb{E}[f_i^{(0)}]) / (n_k \mathbb{E}[f_i^{(k)}]) \le h_2\}$. $\bc_{h_2}$ satisfies that
    \small
    \begin{align*}
		\inf_{k \in \bc_{h_2}} \frac{ n_k ( \mathbb{E}[f_i^{(k)}]  - \Omega_2^{(k)} )}{ n_0 (\mathbb{E}[f_i^{(0)}]  + \Omega_2^{(0)})}  > \sup_{k \in \bc_{h_2}^c} \frac{ n_k ( \mathbb{E}[f_i^{(k)}]  + \Omega_2^{(k)} )}{ n_0 (\mathbb{E}[f_i^{(0)}]  - \Omega_2^{(0)})}.
	\end{align*}
    \normalsize
\end{condition} 

\begin{remark}
    Similar to previous transfer literature \cite{tian2022transfer, zhang2022transfer}, Conditions \ref{cond:delta gap} and \ref{cond:density gap} assume the parameter contrast $\|\bdelta^{(k)}\|_1$ and the population $\tau$-th density $\mathbb{E}[f_i^{(k)}]$ exhibit sufficient theoretical separation between transferable and non-transferable source studies, which exactly satisfies the requirement of the oracle transferable set \eqref{oracle trans set}. To be more specific, the value $\Omega_1^{(k)}$ is the upper error bound of $\| \tdelta^{(k)} - \bdelta^{(k)} \|_1$, while the value $\Omega_2^{(k)}$ is the upper error bound of $|\frac{1}{n_k} \sum_{i=1}^{n_k} \tf_i^{(k)} (\tbeta + \tdelta^{(k)}) - \mathbb{E}[f_i^{(k)}]|$.
\end{remark}

\begin{theorem} \label{thm:detection}
	Assume that Conditions \ref{cond:restricted eigenvalue}--\ref{cond:density gap} hold. Suppose $\lambda_k$ is chosen as in Theorem \ref{thm:convergence} and $s^2 \log p = O(\sqrt{n_k})$ for each study $k$. Then the thresholding parameters $t_1$ and $t_2$ in Algorithm \ref{alm:unify} are chosen as
    \small
    \begin{align*}
        \inf_{k \in \bc_{h_1}^c} \left( \| \bdelta^{(k)} \|_1 - \Omega_1^{(k)} \right) / \sqrt{\frac{\log p}{n_0}}  & \ge t_1 \ge \sup_{k \in \bc_{h_1}} \left( \| \bdelta^{(k)} \|_1 + \Omega_1^{(k)} \right) / \sqrt{\frac{\log p}{n_0}},  \\
        \inf_{k \in \bc_{h_2}} \frac{ n_k ( \mathbb{E}[f_i^{(k)}]  - \Omega_2^{(k)} )}{ n_0 (\mathbb{E}[f_i^{(0)}]  + \Omega_2^{(0)})} & \ge t_2 \ge \sup_{k \in \bc_{h_2}^c} \frac{ n_k ( \mathbb{E}[f_i^{(k)}]  + \Omega_2^{(k)} )}{ n_0 (\mathbb{E}[f_i^{(0)}]  - \Omega_2^{(0)})}.
    \end{align*}
    \normalsize
    Then for the detected transferable set $\wc$ in Algorithm \ref{alm:unify}, for any $\varepsilon > 0$, there exists a positive constant $N(\varepsilon)$ such that when $\min_{k} n_k > N(\varepsilon)$, $\mathbb{P}(\wc = \bc_h) \ge 1 - \varepsilon$.
    Consequently, our practical $\wbeta$ obtained from Algorithm \ref{alm:unify} can enjoy the same $\ell_1/\ell_2$-convergence rates as the oracle estimator $\wbeta_{ora}$ in Corollary \ref{cor:oracle converence rate}.
\end{theorem}

\section{Statistical Inference with Knowledge Transfer} \label{section4}

\subsection{Debiased Estimator based on Neyman Orthogonality}

We now turn to the statistical inference for the target parameter, $\sbeta$, focusing specifically on the first component, $\beta^{*}_1$, without loss of generality. Previous studies \citep{tian2022transfer, li2023estimation, huang2023estimation} addressed this inference task building on the principles of debiased lasso \citep{zhang2014confidence, van2014asymptotically}. Unfortunately, these methods only achieve $\sqrt{n_0}$-normality for the debiased estimators, offering little improvement over debiasing based solely on the target study. The key reason for this limitation is that their debiasing procedures need the covariance matrix to be identical for all samples, a condition that is violated when covariate shift occurs. Also, they do not offer any measures to adjust for the parameter contrast $\bdelta^{(k)}$. The $\sqrt{n_0}$-normality is suboptimal, especially given the wealth of transferable information from source studies. The challenge, then, is how to effectively leverage the source information to enhance inference efficiency, particularly under distributional shift.

To address this challenge, we propose a novel approach motivated by the Neyman orthogonal idea \citep{belloni2019valid}. Denote $\bv_{-1}$ as the last $(p-1)$-dimensions of the vector $\bv \in \mathbb{R}^{p}$. We can rewrite the original model \eqref{qr} as a partially linear quantile regression
\begin{align*} 
	y_i^{(k)}=x_{i,1}^{(k)} \beta^*_1 + (\bx_{i,-1}^{(k)})^\top \bw_{-1}^{(k)} + r_i^{(k)} + \epsilon_i^{(k)}, \quad i=1, \ldots, n_k,
\end{align*}
where $r_i^{(k)}=x_{i,1}^{(k)}\delta^{(k)}_1$ is regarded as an approximation error. A natural moment condition for identifying $\beta^*_1$ is to solve the equation (in $\alpha$)
\begin{align} \label{moment condition}
	\mathbb{E} \sum_{i=1}^{n_k} \left[ \mathbb{I} \left\{y_i^{(k)} \le x_{i,1}^{(k)} \alpha + (\bx_{i,-1}^{(k)})^\top \bw_{-1}^{(k)} + r_i^{(k)} \right\} - \tau \right] x_{i,1}^{(k)} = 0.
\end{align}
The high-dimensional nuisance parameter $\bw_{-1}^{(k)}$ will bring trouble for identifying $\beta^*_1$, since the first-order derivative of \eqref{moment condition} with respect to this nuisance parameter at $\alpha=\beta^*_1$ is non-zero. If conducting valid inference on $\beta^*_1$ directly with moment condition \eqref{moment condition}, the estimator for $\bw_{-1}^{(k)}$ needs to converge faster than $1/\sqrt{n_k}$, which is difficult to achieve. To address this problem, we adopt the Neyman orthogonality to make the moment condition immune to the first-order error from the nuisance estimator. The technique is based on the linear projection of the regressor of interest $x_{i,1}^{(k)}$ on the $\bx^{(k)}_{i,-1}$ weighted by the $\tau$-th density $f_i^{(k)}$:
\begin{align} \label{projection}
	f_i^{(k)}  x_{i,1}^{(k)} = f_i^{(k)}  (\bx^{(k)}_{i,-1} )^{\top} \btheta_0^{(k)} + v_{0i}^{(k)}, \quad i=1, \ldots, n_k,
\end{align}
where the coefficient $\btheta_0^{(k)} \in \arg \min_{\btheta} \mathbb{E} \sum_{i=1}^{n_k} [f_i^{(k)} x_{i,1}^{(k)} - f_i^{(k)} ( \bx^{(k)}_{i,-1} )^{\top} \btheta ]^2$, satisfying the relationship $\mathbb{E}\sum_{i=1}^{n_k}[ f_i^{(k)}  \bx^{(k)}_{i,-1}v_{0i}^{(k)}]=0$. We then construct the orthogonal score function 
\begin{align} \label{score function}
	\psi_i^{(k)}(\alpha, \bw_{-1}^{(k)}, \btheta_0^{(k)}, f_i^{(k)}):=\left[\mathbb{I} \left\{ y_i^{(k)} \le x_{i,1}^{(k)} \alpha + (\bx_{i,-1}^{(k)})^\top \bw_{-1}^{(k)} + r_i^{(k)} \right\} - \tau\right] v_{0i}^{(k)},
\end{align} 
and solve for $\mathbb{E} \sum_{i=1}^{n_k} [ \psi_i^{(k)}(\alpha, \bw_{-1}^{(k)}, \btheta_0^{(k)}, f_i^{(k)})]=0$, where this moment condition can satisfy the orthogonality with respect to high-dimensional nuisance parameters $\bw_{-1}^{(k)}$ and $\btheta_0^{(k)}$, i.e.
\begin{align*}
	\left.  \partial_{\bw} \mathbb{E}  \sum_{i=1}^{n_k} \left[ \psi_i^{(k)}(\beta_1^*, \bw, \btheta_0^{(k)}, f_i^{(k)})\right]  \right| _{\bw=\bw_{-1}^{(k)}} = \left.  \partial_{\btheta} \mathbb{E} \sum_{i=1}^{n_k} \left[ \psi_i^{(k)}(\beta_1^*, \bw_{-1}^{(k)}, \btheta, f_i^{(k)})\right]  \right| _{\btheta=\btheta_0^{(k)}}  = 0.
\end{align*}
Consequently, the estimators for $\bw_{-1}^{(k)}$ and $\btheta_0^{(k)}$ will only introduce a second-order bias, allowing the use of high-dimensional sparse estimates. Note that this approach can avoid the computation of precision matrices and free the identical covariance matrix assumption across studies, which exactly accommodates the circumstance under distribution shift. 

We outline several points here for practical implementation. Firstly, we estimate the $\tau$-th density $f_i^{(k)}$ by $\wf_i^{(k)} = 2\kappa_k / [(\bx_{i}^{(k)})^{\top}\wbeta^{(k)}_{\tau+\kappa_k} - (\bx_{i}^{(k)})^{\top}\wbeta^{(k)}_{\tau-\kappa_k}]$, where $\wbeta^{(k)}_{\tau \pm \kappa_k}$ are computed via $\ell_1$-QR at the $(\tau \pm \kappa_k)$-th quantiles with tuning parameters $\lambda_\tau^{(k)}$ on the $k$-th study respectively. The bandwidth parameter, $\kappa_k$, is set to $\min\{n_k^{-1/6},\tau(1-\tau)/2\}$, as suggested by \cite{koenker2005quantile, giessing2023debiased}.  Next, each high-dimensional projection parameter $\btheta_0^{(k)}$ is determined by $\ell_1$-penalized weighted least squares with tuning parameters $\lambda_{\btheta}^{(k)}$
\begin{align}
	\label{projection estimate}
	\wtheta^{(k)} \in \arg \min_{\btheta} \sum_{i=1}^{n_k} \left[\wf^{(k)}_i x_{i,1}^{(k)} - \wf^{(k)}_i (\bx^{(k)}_{i,-1})^{\top}\btheta \right]^2 + \lambda_{\btheta}^{(k)}\| \btheta \|_1.
\end{align}
For each high-dimensional nuisance parameter $\bw_{-1}^{(k)}$, we calculate the post-selected QR estimator (denoted by $\tw_{-1}^{(k)}$), since the refitting step after feature screening can lead to better finite sample behaviors \citep{belloni2019valid}. Specifically, we screen each source study's covariates by setting threshold for each element of $\wbeta + \tdelta^{(k)}$ obtained in Algorithm \ref{alm:unify}, then apply traditional QR on the selected features. Denote the empirical residuals $\wv_i^{(k)} = \wf^{(k)}_i x_{i,1}^{(k)} - \wf^{(k)}_i (\bx^{(k)}_{i,-1})^{\top} \wtheta^{(k)}$. The empirical score function corresponding to \eqref{score function} is defined as
\begin{align}
    \label{empirical score function}
	\wpsi_i^{(k)} ( \alpha, \tw_{-1}^{(k)}, \wtheta^{(k)}, \wf_i^{(k)}):= \left[ \mathbb{I} \left\{y_i^{(k)} \le x_{i,1}^{(k)} \alpha + (\bx_{i,-1}^{(k)})^\top \tw^{(k)}_{-1} \right\} - \tau \right]  \wv_i^{(k)}.
\end{align}
For joint debiased inference, stronger conditions on parameter and residual shifts are required than those for transfer estimation, as detailed in the next subsection. A natural choice is the detected transferable set $\wc$ from Algorithm \ref{alm:unify}. Since this set is selected for estimation and does not necessarily satisfy the stronger similarity conditions required for valid inference, we further recommend the practical quality-control procedure (Algorithm \ref{alm:quality control}) in Section 8.1.4 of the supplementary materials, which provides an additional diagnostic and allows a fallback to target-only debiasing when the bootstrap-based normality test rejects. Finally, we set the search region $\walpha \in \mathcal{A}_{\wc} := \{ \alpha \in \mathbb{R}:| \alpha - \widehat{\beta}_1| \le 10 [ n_{\wc}^{-1} \sum_{k \in \{0\} \cup \wc} \sum_{i=1}^{n_k} (x^{(k)}_{i,1})^2  ]^{-1/2} / \log n_{\wc} \}$, as suggested by \cite{belloni2019valid}. The whole inference procedure is summarized in Algorithm \ref{alm:inference}. 

\begin{algorithm}[ht]
\small
\caption{Inference for $\beta^{*}_{1}$ with knowledge transfer}
\label{alm:inference}
\KwIn{Detected transferable set $\wc$ and estimators $\wbeta, \{\tdelta^{(k)}\}_{k \in \wc}$ from Algorithm \ref{alm:unify}, \qquad target study $\{((\bx_{i}^{(0)})^{\top}, y_{i}^{(0)})_{i=1}^{n_0}\}$, source studies $\{((\bx_{i}^{(k)})^{\top}, y_{i}^{(k)})_{i=1}^{n_k}\}_{k \in \wc}$, \qquad tuning parameters $\{\lambda_{\tau}^{(k)}, \lambda_{\btheta}^{(k)}\}_{k \in \{0\} \cup \wc}$, threshold parameter $\bar{\lambda}$.}
\begin{enumerate}[itemsep=0pt] 
    \item For $k \in \{0\} \cup \wc$, obtain the density estimate $\wf_{i}^{(k)}$ and conduct the projection \eqref{projection estimate}. 
    \item For $k \in \{0\} \cup \wc$, compute $\ww^{(k)}$ by traditional QR on $\{((\bx_{i,\widehat{S}^{(k)}}^{(k)})^{\top}, y_{i}^{(k)})_{i=1}^{n_k}\}$, where the index set $\widehat{S}^{(k)}=\{j: |\widehat{\beta}_j + \widetilde{\delta}^{(k)}_j| > \bar{\lambda} \}$, then enlarge $\ww^{(k)}$ (with zero) to $\widetilde{\bw}^{(k)} \in \mathbb{R}^{p}$.
    \item Incorporate empirical score functions \eqref{empirical score function} and solve for
    \begin{align*}
        \walpha \in \underset{\alpha \in \mathcal{A}_{\wc}}{\arg \min} 
        \left| \sum_{k \in \{0\} \cup \wc} \sum_{i=1}^{n_k} \wpsi_i^{(k)}(\alpha, \tw_{-1}^{(k)}, \wtheta^{(k)}, \wf_i^{(k)}) \right|.
    \end{align*}
\end{enumerate}
\KwOut{Debiased estimator $\walpha$.}
\end{algorithm}

\subsection{Asymptotic Normality for Debiased Estimator}

As in Theorem \ref{thm:convergence}, we consider an arbitrary prespecified source set $\bc$. Compared with the preceding estimation and detection results, valid joint inference requires additional assumptions, particularly stronger control over source-to-target differences in the coefficient of interest. It is worth noting that consistent detection of the oracle transferable set does not, by itself, guarantee these inference-specific conditions. Let $c, C$ be generic constants which can be different at different places.

\begin{condition} \label{cond:moment conditions} 
	For each study $k \in \{0\} \cup \bc$, (i) $\mathbb{E} [ (\xi^\top \bx^{(k)}_{i,-1})^2 (r_i^{(k)})^2] \le C \| \xi \|_2^2 \mathbb{E} [(r_i^{(k)})^2]$, for any $\xi \in \mathbb{R}^{p-1}$. (ii) $\sum_{k \in \bc} \frac{n_k}{\nc} \mathbb{E}[(r_i^{(k)})^2] = o(1/\sqrt{\nc})$. (iii) $\sum_{k \in \bc} \frac{n_k}{\nc} |\mathbb{E}[f_i^{(k)} v_{0i}^{(k)} r_i^{(k)}]| = o(1/\sqrt{\nc})$.
\end{condition}

\begin{condition} \label{cond:weighted lasso} 
    For each study $ k \in \{0\} \cup \bc$, (i) $\| \btheta_0^{(k)} \|_2 \le C$; (ii) There exist sparsity levels $\{ s_k^* \}_{k \in \{0\} \cup \bc}$ and vectors $\{\btheta^{(k)}\}_{k \in \{0\} \cup \bc}$ such that $s_k^*=\max \{s + s_k, \|\btheta^{(k)}\|_0\}$, where each $\btheta^{(k)}$ satisfies $(\bx_{i,-1}^{(k)})^\top \btheta_0^{(k)} = (\bx_{i,-1}^{(k)})^\top \btheta^{(k)} + r_{\btheta i}^{(k)}$ with $ \| \btheta_0^{(k)} - \btheta^{(k)} \|_1 = C s_k^* \sqrt{\log p / n_k}$, $\mathbb{E}[(r_{\btheta i}^{(k)})^2] \le C s_k^* / n_k$; (iii) $c \le \min_{1 < j \le p}\mathbb{E} [ (f_i^{(k)} x^{(k)}_{i,j}v_{0i}^{(k)})^2 ]^{1/2} \le \max_{1 < j \le p}\mathbb{E} [(f_i^{(k)} x^{(k)}_{i,j}v_{0i}^{(k)})^3]^{1/3} \le C$. 
\end{condition}

\begin{condition} \label{cond:density estimation}
	For each study $ k \in \{0\} \cup \bc$, (i) Conditional density functions $f(\cdot|\bx_i^{(k)})$ are uniformly lower bounded by a positive constant; (ii) For $u = \tau \pm \kappa_k$, assume that $u$-quantile$(y_i^{(k)}|\bx_i^{(k)}) = (\bx_i^{(k)})^\top \bw^{(k)}_u + \epsilon_{ui}^{(k)}$, where $\mathbb{E}[\epsilon_{ui}^{(k)}] \le o( n_k^{-1/2} )$, $| \epsilon_{ui}^{(k)} | \le o(\kappa_k)$, $\| \bw^{(k)}_u \|_0 \le C s_k^*$ and $\max_i \| \bx_i^{(k)} \|_\infty \sqrt{(s_k^*)^2 \log p / n_k} = o(\kappa_k)$.
\end{condition}

Conditions \ref{cond:moment conditions}-\ref{cond:density estimation} generalize the assumptions of single-study version in \cite{belloni2019valid} to the transfer setting. Condition \ref{cond:moment conditions} set up the moment requirements for approximation error $r_i^{(k)}$, where each source's quantity $r_i^{(k)}$ is linked to the whole sample size $\nc$. Condition \ref{cond:weighted lasso} follows the original settings, providing the necessary requirements for establishing convergence rates of weighted Lasso estimators $\{\wtheta^{(k)}\}_{k \in \{0\} \cup \bc}$. Condition \ref{cond:density estimation} is for the estimation of the conditional density at $(\tau \pm \kappa_k)$-th quantiles, where the approximately sparse QR model is assumed around the quantile index $\tau$.  

\begin{remark}
    Since our inference framework performs debiasing using source studies, it requires stronger similarity in the coefficient of interest than transfer estimation in Theorem \ref{thm:convergence}. Specifically, Condition \ref{cond:moment conditions} (iii) requires the approximation bias $\sum_{k\in\bc} \frac{n_k}{n_{\bc}} |E[f_i^{(k)}v_{0i}^{(k)}r_i^{(k)}]|$ to be $o(n_{\bc}^{-1/2})$, where $E[f_i^{(k)}v_{0i}^{(k)}r_i^{(k)}]=E[(v_{0i}^{(k)})^2]\delta_1^{(k)}$. Nevertheless, this assumption is reasonable to hold in many practical scenarios. For instance, in multi-source genetic studies, researchers are often interested in estimating the common effect of a specific genetic variant on a phenotype across different sources \citep{michailidou2017association, kunkle2019genetic}.
\end{remark}

\begin{theorem} \label{thm:inference}
	Assume Conditions \ref{cond:moment conditions}-\ref{cond:density estimation} and all conditions required by Theorem \ref{thm:convergence}. For each study $k \in \{0\} \cup \bc$, let $\lambda_\tau^{(k)} = c_\lambda \sqrt{n_k \log p}$, where $c_\lambda$ is a large enough constant. If the parameters satisfy the following growth conditions
    
	\small
    \begin{equation} \label{Theorem 4 assumption}
	\begin{aligned}
		& \frac{s \log p}{\nc} + \sum_{k \in \mathcal{C}} \frac{n_k}{\nc} \sqrt{\frac{\log p}{n_k \wedge n_0}} \| \bdelta^{(k)}\|_1= o(1/\sqrt{\nc}), \\
		& \sum_{k \in \{0\} \cup \bc} \frac{n_k}{\nc} \left[ \frac{1}{\kappa_k} \sqrt{\frac{s_k^* \log p}{n_k}} + (\kappa_k)^2 +  \frac{\lambda_{\btheta}^{(k)} s_k^*}{n_k} \right]= o(\nc^{-1/4}), \\
		&  \max_{k \in \{0\} \cup \bc} \left\{ s_k^* + \frac{n_k s_k^* \log p}{ (\kappa_k\lambda^{(k)}_{\btheta})^2} + \left( \frac{n_k (\kappa_k)^2}{\lambda^{(k)}_{\btheta}} \right)^2 \right\} \frac{\log p}{\nc} = o(1/\sqrt{\nc}),
	\end{aligned}
    \end{equation}
	\normalsize
	then for the estimator $\walpha$ obtained from Algorithm \ref{alm:inference} with sources in the index set $\bc$, we have
    \small
    \begin{equation*}
	\begin{aligned}
		& \qquad \qquad \qquad \sigma_{\bc}^{-1} \sqrt{\nc} (\walpha - \beta_1^*) = \mathbb{U}_{\bc}(\tau) + o_p (1), \qquad \mathbb{U}_{\bc}(\tau)  \leadsto \mathcal{N}(0,1), \\
        & \text{where} \quad \mathbb{U}_{\bc}(\tau)  = \left( \tau (1 - \tau) \sum_{k \in \{0\} \cup \bc} \frac{n_k}{n_{\bc}} \mathbb{E}[(v_{0i}^{(k)})^2] \right)^{-1/2} \cdot \frac{1}{\sqrt{\nc}} \sum_{k \in \{0\} \cup \bc} \sum_{i=1}^{n_k} \psi_i^{(k)}(\alpha, \bw_{-1}^{(k)}, \btheta_0^{(k)}, f_i^{(k)}), \\
        & \text{and} \quad \sigma_{\bc}^{2} = \left(\sum_{k \in \{0\} \cup \bc} \frac{n_k}{n_{\bc}} \mathbb{E}[f_i^{(k)}x_{i,1}^{(k)}v_{0i}^{(k)}] \right)^{-2} \cdot \sum_{k \in \{0\} \cup \bc} \frac{n_k}{n_{\bc}} \mathbb{E} [\tau(1-\tau)(v_{0i}^{(k)})^2] .
	\end{aligned}
    \end{equation*}
    \normalsize
	Moreover, the estimator $\widehat{\sigma}_{\bc}^{2} = \tau (1 - \tau) n_\bc [\sum_{k \in \{0\} \cup \bc} \sum_{i=1}^{n_k} (\wv_{i}^{(k)})^2 ]^{-1}$ is consistent.
\end{theorem} 

\begin{remark}
    As implied by the minimax lower bound in \cite{li2023transfer}, the $\sqrt{n_{\bc}}$-rate becomes theoretically possible only when the similarity conditions are linked to the total sample size $n_{\bc}$ instead of $n_0$ only. As shown in the parameter growth conditions stated in \eqref{Theorem 4 assumption}, we relate the parameter growth of each individual study to the total sample size $n_{\bc}$. Owing to the Neyman orthogonal technique we employ, these requirements are not stringent, where a sufficient condition is to let $\kappa_k \asymp n_k^{-1/6}, \lambda_{\btheta}^{(k)} \asymp \sqrt{n_k \log p }, s_k^* \log p = o(n_k^{2/3} / \sqrt{n_{\bc}}), \| \bdelta^{(k)} \|_1 \log p  = o(\sqrt{n_k \wedge n_0} / \sqrt{n_{\bc}})$ for each study $k \in \{0\} \cup \bc$. 
\end{remark}

With the asymptotic normality of $\walpha$ and consistent variance estimators $\widehat{\sigma}_{\bc}^2$ provided in Theorem \ref{thm:inference}, we can construct valid confidence intervals and hypothesis test procedures for the objective component $\beta^*_1$. Now we give some explanations for results in Theorem \ref{thm:inference}. Firstly, our estimator $\walpha$ enjoys $\sqrt{\nc}$-normality, which is much sharper than the previous results, demonstrating the effectiveness of our debiased transfer methods. Secondly, conditions of Theorem \ref{thm:inference} allow for nonidentical covariate distributions across studies, as well as robust to the divergence of covariance matrices between target and source studies, demonstrating our methods' superiority under covariate shift. Moreover, note that the asymptotic variance $\sigma_{\bc}^2$ includes $v_{0i}^{(k)}$ that is inversely proportional to the $\tau$-th density $f_i^{(k)}$, clearly revealing that residual shift may increase the variance and deteriorate the inference efficiency, and that's why we choose to only debias on the detected transferable set. 

\section{Simulation} \label{section5}

\subsection{Estimation Errors}

We follow a similar setup as in \cite{zhang2022transfer, huang2023estimation}. Consider a target study with size $n_0=200$, dimension $p=500$, and sparsity level $s=10$. We simulate two types of target model:
\small
\begin{align} \label{simu model}
	\text{Homoscedastic: }  y_i^{(0)} = ( \bx_i^{(0)} )^{\top} \sbeta + \epsilon_i^{(0)},  ~
	\text{Heteroscedastic: } y_i^{(0)} = ( \bx_i^{(0)} )^{\top} \sbeta + |x_{i,1}^{(0)}| \epsilon_i^{(0)},
\end{align}
\normalsize
where $\sbeta = (\mathbf{1}_s, \mathbf{0}_{p-s})^\top$, $\bx_i^{(0)} \sim \mathcal{N}(\mathbf{0}_p, \boldsymbol{\Sigma}_{\bx})$ with $\boldsymbol{\Sigma}_{\bx} = (0.7^{|i-j|})_{1 \leq i, j \leq p}$, and $\epsilon_i^{(0)}$ follows the shifted standard normal distribution with quantile levels $\tau \in \{0.2, 0.5, 0.7\}$. The source study number $K=5$ and follows the homoscedastic/heteroscedastic model type \eqref{simu model} with $(y_i^{(k)}, \bx_i^{(k)}, \bw^{(k)}, \epsilon_i^{(k)})_{i=1}^{n_k}$. To simulate the distribution shift, we consider the following parameter configurations for the source studies. 

\begin{itemize}[leftmargin=*]
    \item \textbf{Covariate shift.} For each source study, we set the covariates $\{\bx_i^{(k)}\}_{i=1}^{n_k}$ follow $\mathcal{N} ( \mathbf{0}_p, \boldsymbol{\Sigma}_{\bx} + \boldsymbol{\epsilon}_{\bx}^{(k)} (\boldsymbol{\epsilon}_{\bx}^{(k)})^\top)$ with $\boldsymbol{\epsilon}_{\bx}^{(k)} \sim \mathcal{N}(\mathbf{0}_p, 0.3^2 \boldsymbol{I}_p)$, then normalize their variance to one. 
    \item \textbf{Parameter shift.} For each source study, denote $\{\zeta_j^{(k)}\}_{j=1}^p$ as independent Rademacher random variables (taking values in $\{1,-1\}$ with equal probability). Following the notation in Section \ref{section3}, we denote the first oracle set $\bc_{h_1}$, in which the $\ell_1$-norm of parameter contrast $\|\bdelta^{(k)}\|_1 = \| \bw^{(k)} - \sbeta \|_1 $ of each source study is relatively small. We set
    \small
	\begin{align*}
	    w_j^{(k)}= \begin{cases} \beta^*_j + h_1 / 100 \cdot \zeta_j^{(k)} \cdot \mathbb{I}\left\{j \in \kappa_k \cup [s/2]\right\}, & k \in \bc_{h_1}, \\ h_1 / 10 \cdot \zeta_j^{(k)} \cdot \mathbb{I}\left\{j \in \kappa_k \cup[s/2]\right\}, & k \in \bc_{h_1}^c,\end{cases}
	\end{align*}
    \normalsize
	where $[s/2] = \{1, \ldots, s/2\}$, $\kappa_k$ is a random subset chosen from $\{s/2+1, \ldots, p\}$ with $|\kappa_k|=50$. It can be verified that $\max_{k \in \bc_{h_1}}\|\bdelta^{(k)}\|_1 < h_1 < \min_{k \in \bc_{h_1}^c}\|\bdelta^{(k)}\|_1$. In this setting, we consider $h_1=5$. 
	\item \textbf{Residual shift.} For each source study $k=1, \ldots, K$, the residual $\epsilon_i^{(k)}$ either follows the standard normal distribution, or it follows one type of the following abnormal distributions: 1) (Cauchy) cauchy distribution: $\mathcal{C}(0,3)$; 2) (Mixed) mixed gaussian: $z\mathcal{N}(-3,0.5) + (1-z)\mathcal{N}(3,0.5)$, where $z \sim \mathcal{B}ernoulli(\tau)$; 3) (Noisy) noisy setting: $\mathcal{N}(0,5^2)$. All distributions have already been shifted to meet $\mathbb{P}(\epsilon_i^{(k)} \le 0|\bx_i^{(k)})=\tau$. We denote the set containing all standard normal source studies by $\bc_{h_2}$, where we can let $h_2=4$ satisfy $\max_{k \in \bc_{h_2}} \mathbb{E}[f_i^{(0)}] / \mathbb{E}[f_i^{(k)}] < h_2 < \min_{k \in \bc_{h_2}^c} \mathbb{E}[f_i^{(0)}] / \mathbb{E}[f_i^{(k)}]$. To clearly demonstrate the negative impact of residual shift, we set the sample size to $n_k = 100 \cdot \mathbb{I}\{k \in \bc_{h_2}\} + 200 \cdot \mathbb{I}\{k \in \bc_{h_2}^c\}$. 
\end{itemize}

Based on the definition in \eqref{oracle trans set}, we define $\bc_{h} = \bc_{h_1} \cap \bc_{h_2}$ as the transferable oracle set. To capture both parameter and residual shift in one setting, we vary the parametric transferable set $\bc_{h_1}$ from 1 to 5, and assume that each source study follows a divergent distribution with probability 1/2. For example, in the ``Cauchy'' setting, each $\epsilon_i^{(k)}$ either follows a normal distribution like $\epsilon_i^{(0)}$ or a Cauchy distribution with equal probability. The same applies to the ``Mixed'' and ``Noisy'' settings. We compare six specific methods, which are: 1) \emph{L1-QR}: conduct $\ell_1$ quantile regression only on the target data (baseline); 2) \emph{Pooling}: Algorithm \ref{alm:unify} (without detection) on all source studies; 3) \emph{Oracle}: Algorithm \ref{alm:unify} (without detection) on the oracle transferable set $\bc_{h}$; 4) \emph{TransQR}: Algorithm \ref{alm:unify}, our transfer approach; 5) \emph{Oracle\_PS}: Algorithm \ref{alm:unify} (without detection) on $\bc_{h_1}$, the oracle choice in the previous literature; 6) \emph{TransQR\_2step}: two-step transfer approach proposed in \cite{huang2023estimation}. The choices of all tuning parameters are provided in the supplementary materials.  

We evaluate the performance of each method in terms of the $\ell_2$-error $\| \wbeta - \bbeta \|_2^2$. We repeat the experiment for 100 times and average the results for each setting respectively. The average $\ell_2$-error of six methods under homoscedastic model are shown in Figure \ref{simu:est_homo_normal}. More extensive simulation results, including scenarios with strong covariate shift and heavy-tailed target distributions, are provided in Section 8.2 of supplementary materials.

\begin{figure}[ht!] 
    \begin{center}
        \includegraphics[width=0.7\linewidth]{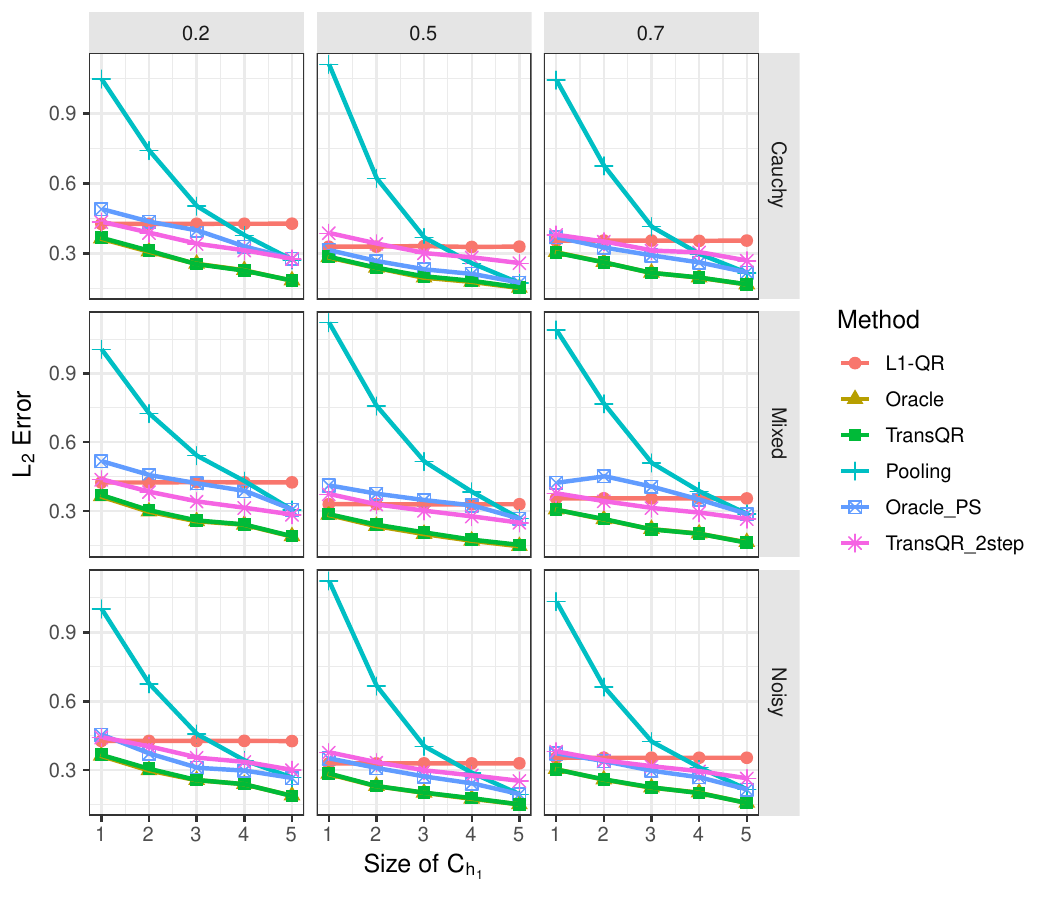}
    \end{center}
    \caption{Average $\ell_2$-error of various methods under homoscedastic model.}
    \label{simu:est_homo_normal}
\end{figure}

From Figure \ref{simu:est_homo_normal}, we highlight three points. First, as expected, the \emph{Oracle} approach achieves the best performance across all scenarios, whereas the \emph{Pooling} approach performs the worst when transferable source studies are scarce. This finding aligns with our theoretical motivation that knowledge transfer should rely on informative sources; otherwise, negative transfer may occur. Second, the \emph{Oracle\_PS} method underperforms relative to the \emph{Oracle} method across all settings, indicating that residual shift indeed affects transfer performance. The only difference between the two methods is that \emph{Oracle\_PS} accounts solely for parameter shift and includes sources with heterogeneous residual distributions. The superior performance of \emph{Oracle} confirms that our proposed metric $n_k \mathbb{E}[f_i^{(k)}]$ serves as a crucial measure for controlling residual shift. Finally, it is noteworthy that our proposed method, \emph{TransQR}, performs almost identically to the oracle benchmark, demonstrating the effectiveness of our transferable set detection under distribution shift. In contrast, the two-step approach \emph{TransQR\_2step} performs unsatisfactorily, revealing its limited ability to handle residual shift and its sensitivity to covariate shift. 

\subsection{Inference Results} \label{section5.2}

To give some intuitive understandings of Algorithm \ref{alm:inference}, we conduct simulations for statistical inference on $\beta^*_1$ under the heteroscedastic model setting in \eqref{simu model} with some adaptive changes. We let $|\bc_{h_1}|=5, h_1=10$ with $|\delta_1^{(k)}|=0.01$ (controlling approximation error $r_i^{(k)}$), also let $n_0=n_{k | k \in \bc_{h_2}}=300$ and $n_{k | k \in \bc_{h_2}^c}=500$. To focus on the performance of the inference procedure and reduce computational complexity, we assume that $\bc_{h_2}:=\{1,2,3\}$ is known and that the source studies in $\bc_{h_2}^c$ follow the ``Mixed'' setting.

We compare five specific methods, including: 1) \emph{Debias\_tar}: debiasing on the target study by Algorithm \ref{alm:inference} with the \emph{L1-QR} estimate (baseline); 2) \emph{Debias\_trans\_pool}: debiasing on all studies by Algorithm \ref{alm:inference} with the \emph{Pooling} estimate; 3) \emph{Debias\_trans\_tar}: debiasing on the target study by Algorithm \ref{alm:inference} with the \emph{TransQR} estimate; 4) \emph{Debias\_trans}: debiasing on the target and transferable source studies by Algorithm \ref{alm:inference} with the \emph{TransQR} estimate; 5) \emph{Debias\_dl}: debiased Lasso approach proposed in \cite{huang2023estimation} with the \emph{TransQR\_2step} estimate. Note that \emph{Debias\_trans} is our primary approach for comparison. The experiment was repeated 1,000 times at $\tau \in \{0.2, 0.5, 0.7\}$. The choices of all tuning parameters are provided in the supplementary materials. We evaluate the inference efficiency via multiple dimensions. Table \ref{tab:infer_normal} reports the coverage probabilities of 95\% confidence intervals, the average absolute bias, the empirical standard error, and the average estimated standard error. Figure \ref{simu:infer_normal} further illustrates these results through boxplots of bias, boxplots of confidence interval lengths, and density plots of normalized estimates compared with $\mathcal{N}(0,1)$. 

\begin{table}[ht]
\footnotesize
\centering
\renewcommand{\arraystretch}{0.8}
\setlength{\tabcolsep}{2.5pt}
\captionsetup{skip=3pt}
\caption{Coverage probability of $95\%$ confidence intervals (CP), average absolute bias (Bias), standard error (SE), and average estimated standard error (ESE) of various methods.}
\label{tab:infer_normal}
\begin{tabular}{@{}l *{12}{c} @{}}
    \toprule
    \multirow{2}{*}{\normalsize{Method}} & \multicolumn{4}{c}{$\tau=0.2$} & \multicolumn{4}{c}{$\tau=0.5$} & \multicolumn{4}{c}{$\tau=0.7$} \\
    \cmidrule(lr){2-5} \cmidrule(lr){6-9} \cmidrule(l){10-13}
    & CP & Bias & SE & ESE & CP & Bias & SE & ESE & CP & Bias & SE & ESE \\
    \midrule
    \emph{Debias\_tar} & 0.965 & 0.101 & 0.123 & 0.138 & 0.954 & 0.0857 & 0.108 & 0.110 & 0.953 & 0.0880 & 0.110 & 0.117 \\
    \emph{Debias\_trans\_pool} & 0.906 & 0.0654 & 0.0812 & 0.0735 & 0.930 & 0.0543 & 0.0666 & 0.0627 & 0.913 & 0.0598 & 0.0704 & 0.0657 \\
    \emph{Debias\_trans\_tar} & 0.960 & 0.0966 & 0.123 & 0.138 & 0.961 & 0.0830 & 0.103 & 0.110 & 0.957 & 0.0867 & 0.109 & 0.117 \\
    \emph{Debias\_trans} & 0.948 & 0.0625 & 0.0755 & 0.0804 & 0.960 & 0.0522 & 0.0642 & 0.0677 & 0.948 & 0.0572 & 0.0683 & 0.0710 \\
    \emph{Debias\_dl} & 0.936 & 0.0865 & 0.109 & 0.104 & 0.939 & 0.0695 & 0.0874 & 0.0876 & 0.936 & 0.0756 & 0.0942 & 0.0910 \\
    \bottomrule
\end{tabular}
\end{table}

\begin{figure}[!ht] 
\begin{center}
    \includegraphics[width=0.65\linewidth]{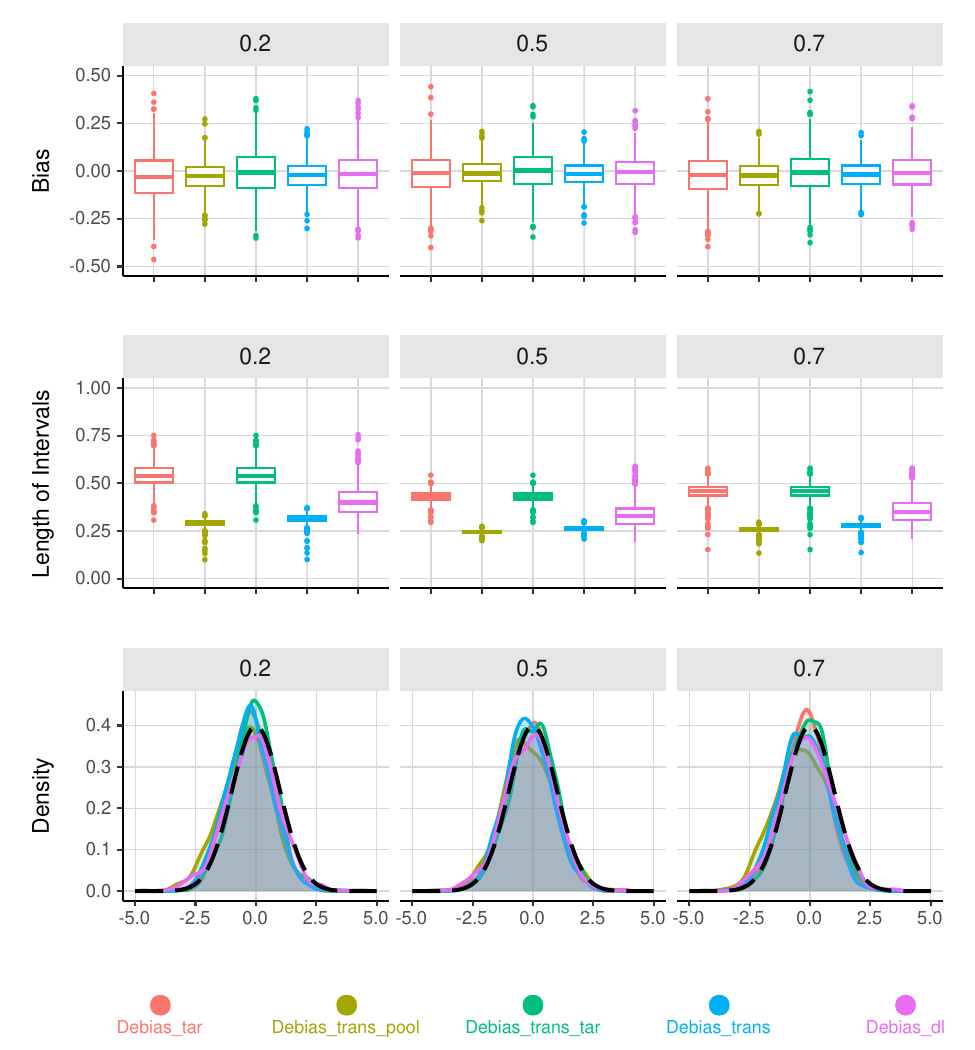}
\end{center}
\caption{Inference performance of various methods, including boxplots of bias, boxplots of confidence interval lengths, and density plots of normalized estimates compared to $\mathcal{N}(0,1)$ with black dotted line.}
\label{simu:infer_normal}
\end{figure} 

From Table \ref{tab:infer_normal} and Figure \ref{simu:infer_normal}, we highlight three points. Firstly, our method \emph{Debias\_trans} achieves the better coverage performance, with coverage probabilities closest to 0.95 among all methods. Because \emph{Debias\_tar}, \emph{Debias\_trans\_tar} and \emph{Debias\_dl} only debias on the target study, their coverage probabilities are somewhat a little higher/lower than $0.95$. Moreover, \emph{Debias\_trans} exhibits smaller estimation bias and shorter confidence intervals than these target-only methods, consistent with our theory that incorporating informative sources yields faster asymptotic rates. Secondly, the \emph{Debias\_trans\_pool} approach exhibits the poorest coverage performance. In principle, \emph{Debias\_trans\_pool} should enjoy reduced bias and tighter intervals due to its larger effective sample size (1000 more observations than \emph{Debias\_trans} in our simulation setting). However, residual shift undermines these potential gains: both its bias and interval lengths are only comparable to those of \emph{Debias\_trans}, while its coverage is noticeably lower. A likely explanation is the downward bias in its estimated standard errors. When \emph{Debias\_trans\_pool} includes source studies with large residual shift, the individual density $f_i^{(k)}$ for $k \in \bc_{h_2}^c$ is very small, which causes instability in estimation (often upward empirically). Since the standard error estimator has $\wf_i^{(k)}$ in the denominator, this results in underestimation of the variance. Lastly, comparing the density plots with the standard normal curve, the normalized estimates from most methods align well with the standard normal distribution. In contrast, \emph{Debias\_trans\_pool} shows a slight but obvious deviation, further indicating the abnormality induced by pooling source studies under residual shift.

\section{Application} \label{section6}

The proposed transfer learning algorithm is applied to the Genotype-Tissue Expression (GTEx) data, available at \href{https://gtexportal.org/}{https://gtexportal.org/}. We use the GTEx Analysis V8 ``Gene read counts by tissue'' dataset, which contains 54 tissues, 948 individuals, and 58,219 genes, including 13 brain tissues and 41 non-brain tissues. Following the previous work \citep{li2022transfer, zhang2022transfer, zhang2024concert}, our study examines gene regulation within the central nervous system (CNS) across various tissues. The CNS-related genes are grouped under MODULE\_137, which includes 545 genes, along with an additional 1,632 genes that show significant enrichment in the same experiments as those in the module. Please refer to \href{https://www.gsea-msigdb.org/gsea/msigdb/cards/MODULE_137.html}{https://www.gsea-msigdb.org/gsea/msigdb/cards/MODULE\_137.html} for a detailed description.

Our analysis focuses on the lower-quantile expression levels of the genes JAM2 and SH2D2A. Specifically, we aim to predict their expression at $\tau = 0.2$ within a given target tissue by leveraging the expression of other central nervous system (CNS) genes. Both JAM2 and SH2D2A are known to play important roles in cellular functions and disease progression, particularly in cancer and immune response. We consider 11 brain tissues as separate target studies, 39 non-brain tissues as source studies. For each target tissue, the data are divided into five folds: one fold is used as the prediction set, while the remaining four folds are used for training. Details of the dataset, gene functions, and data analysis procedures are provided in the supplementary materials. Figure \ref{fig:realdata} reports the relative prediction errors of \emph{Pooling}, \emph{TransQR\_2step}, and our proposed method \emph{TransQR}, all compared with the baseline \emph{L1-QR}. The black horizontal line marks a ratio of 1. Consistent with the simulation results, the \emph{Pooling} approach achieves the least improvement across all tissues for both genes, indicating the presence of non-transferable tissues whose inclusion may deteriorate prediction performance. By contrast, \emph{TransQR} achieves the lowest prediction errors in almost every tissue for both genes (except C.hemisphere and N.C.A.B.ganglia for JAM2), demonstrating that our transferable set detection can effectively identify informative studies under various types of distribution shift. The two-step transfer method, \emph{TransQR\_2step}, generally lies between \emph{Pooling} and \emph{TransQR}, further underscoring that it is not sufficient to account only for parameter shift; a robust and efficient transfer framework must also address covariate shift and residual shift. Compared with mean regression results reported in \cite{li2022transfer, Zhao2026Residual}, our findings provide richer insights into the full distribution of gene expression, particularly at the lower quantile.

\begin{figure}[!ht] 
\begin{center}
    \includegraphics[width=0.75\linewidth]{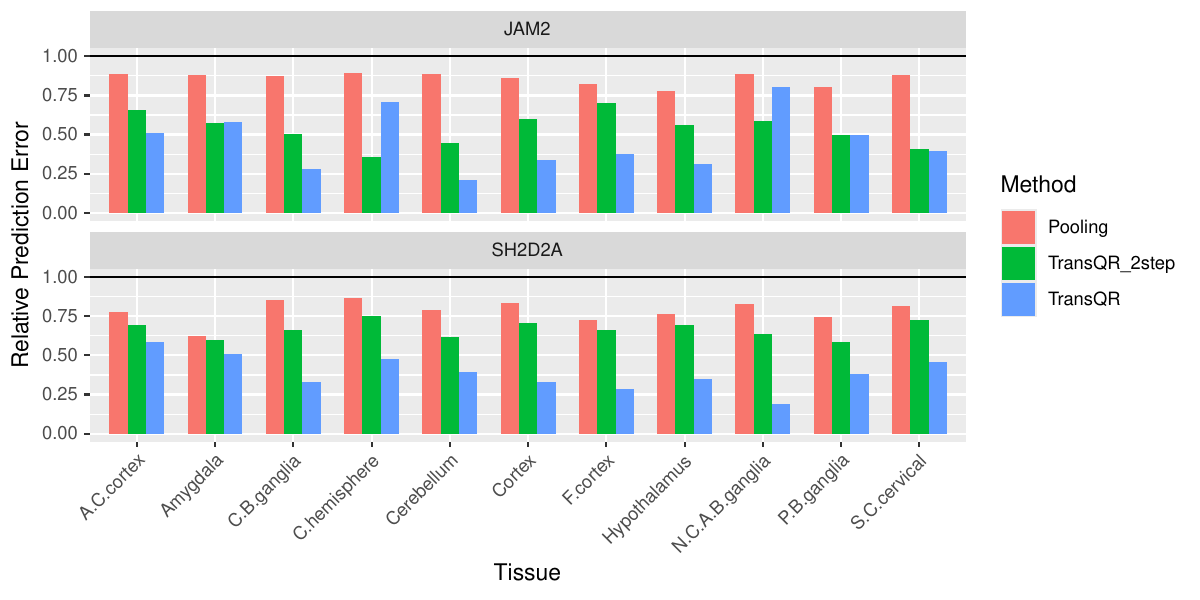}
\end{center}
\caption{Relative prediction errors for the gene expression levels of JAM2 and SH2D2A at $\tau=0.2$.}
\label{fig:realdata}
\end{figure} 

\section{Discussion}

There are several avenues to generalize our method. One interesting direction is to replace the current hard-thresholding transfer strategy with a weighted transfer framework \cite{Zhao2026Residual}, where source studies contribute according to their information quality. Such a strategy may improve the utilization of partially transferable sources and potentially further enhance estimation efficiency. However, designing an effective weighting scheme for quantile transfer learning remains challenging, as the weight should simultaneously account for both the local information quality (reflected by the $\tau$-th density) and the parameter shift across studies to properly balance variance reduction and transfer bias. We leave the development and theoretical analysis of such shift-aware weighting strategies for future research.


    
Furthermore, our algorithm can be easily generalized to the communication-efficient version with multiple data storage nodes \citep{chen2020distributed}. By carefully observing our Algorithm \ref{alm:unify}, almost every step is based on the local data and the summary statistics from other studies. We just need to transmit the estimators like $\tbeta, \{\tdelta^{(k)}\}_{k=1}^{K}$ between studies, and conduct distributed high-dimensional quantile regression \citep{chen2020distributed}.  

Lastly, distributionally robust optimization \citep{gao2023distributionally} aims to obtain an optimized target distribution within a feasible distribution set. We can generalize this idea to transfer learning \citep{xiong2023distributionally}, which may help us to select source domains more effectively and learn a robust target model under distribution shift. We leave these directions for future work.

\clearpage

\section{Supplementary Materials}

\vspace{-3mm}

\startcontents[supple]
\printcontents[supple]{l}{2}{\centering \textbf{\large Outline}}

This appendix is organized into three subsections. Section \ref{section-sup1} offers additional explanations for the main text, including the reasonability of our detection rules under covariate shift, the algorithm of the constrained $\ell_1$-minimization problem \eqref{l1 transfer}, the practical guide for selecting hyperparameters and the discussion of similarity conditions for statistical inference. Section \ref{section-sup2} presents extensive simulation results under various settings. Section \ref{section-sup3} provides details on the real data application, including the dataset description, the biological context and the implementation procedure.

\subsection{Additional Explanations for the Main Text} \label{section-sup1}

\subsubsection{Reasonability of Detection Rules under Covariate Shift}

In Section \ref{section2-1}, we introduced the basic idea of excluding source studies with small $\tau$-th densities under residual shift using a simplified setting of low-dimensional linear/quantile regression with identical covariance matrices. However, one may naturally question whether our proposed definition of an oracle transferable set remains valid under covariate shift—that is, when covariance matrices differ across studies. To address this concern clearly, we now present a more sophisticated example to illustrate the underlying rationale.

Compared to the low-dimensional linear setting in Section \ref{section2-1}, here we assume different covariance matrices $\Sigma_0 := \mathbb{E}[\bx_i^{(0)}(\bx_i^{(0)})^\top], \Sigma_1 := \mathbb{E}[\bx_i^{(1)}(\bx_i^{(1)})^\top]$ for the target and source study. Based on the knowledge of linear models, the target and pooling OLS estimator $\wbeta_{tar}, \wbeta_{pool}$ satisfy: \small
\begin{align*}
    \left( \wbeta_{tar} - \sbeta \right) \xrightarrow{d} \mathcal{N} \left(0, \frac{\sigma_0^2}{n_0} \Sigma_{0}^{-1} \right), \quad 
    \left(\wbeta_{pool} - \sbeta \right) \xrightarrow{d} \mathcal{N}\left(0, \frac{\bar{\Sigma}^{-1} \bar{\Sigma}_\sigma \bar{\Sigma}^{-1}}{n_0 + n_1} \right),
\end{align*}
\normalsize
where $\bar{\Sigma} = (n_0 \Sigma_0 + n_1 \Sigma_1) / (n_0 + n_1), \bar{\Sigma}_\sigma = (n_0 \sigma_0^2 \Sigma_0 + n_1 \sigma_1^2 \Sigma_1) / (n_0 + n_1)$. Suppose we are interested in the asymptotic variance of the $j$-th component (using symbol $\mathrm{Avar}_j [\cdot]$), we have that 
\small
\begin{align*}
    \mathrm{Avar}_j [ \wbeta_{tar} ] = \frac{\sigma_0^2}{n_0} (\Sigma_0^{-1})_{jj}, \quad \mathrm{Avar}_j [ \wbeta_{pool} ] = \frac{\sigma_0^2}{n_0 + n_1} (\bar{\Sigma}^{-1})_{jj} + \frac{n_1 (\sigma_1^2 - \sigma_0^2)}{(n_0 + n_1)^2} (\bar{\Sigma}^{-1} \Sigma_1 \bar{\Sigma}^{-1})_{jj}.
\end{align*}
\normalsize
If the residual variance $\sigma_1^2$ is much larger than $\sigma_0^2$, for example when $\sigma_1^2 \ge \sigma_0^2 + \frac{\sigma_0^2 (n_0 + n_1)^2 (\Sigma_0^{-1})_{jj}}{n_0 n_1 (\bar{\Sigma}^{-1} \Sigma_1 \bar{\Sigma}^{-1})_{jj}}$, the pooling estimator will be much less efficient compared to the target one. Moving to the quantile regression, we can similarly get that
\begin{align*}
    & \left( \wbeta_{tar} - \sbeta \right) \xrightarrow{d} \mathcal{N} \left(0, \frac{\tau(1-\tau)}{n_0(f^{(0)}(0))^2} \Sigma_0^{-1}\right), \quad \left( \wbeta_{pool} - \sbeta \right) \xrightarrow{d} \mathcal{N}\left(0, \frac{\tau(1-\tau)}{n_0 + n_1} \Sigma_f^{-1} \bar{\Sigma} \Sigma_f^{-1} \right), \\
    & \mathrm{Avar}_j [ \wbeta_{tar} ] = \frac{\tau(1-\tau)}{n_0(f^{(0)}(0))^2} (\Sigma_0^{-1})_{jj}, \\
    & \mathrm{Avar}_j [ \wbeta_{pool} ] = \frac{\tau(1-\tau)}{f^{(0)}(0)} \left[ \frac{(\bar{\Sigma}_f^{-1})_{jj}}{n_0 + n_1}  + \frac{n_1 (f^{(0)}(0) - f^{(1)}(0))}{(n_0 + n_1)^2} (\bar{\Sigma}_f^{-1} \Sigma_1 \bar{\Sigma}_f^{-1})_{jj} \right],
\end{align*}
where $\bar{\Sigma}_f = (n_0 f^{(0)}(0) \Sigma_0 + n_1 f^{(1)}(0) \Sigma_1) / (n_0 + n_1)$. If $f^{(1)}(0) \le f^{(0)}(0) - \frac{(n_0 + n_1)^2 (\Sigma_0^{-1})_{jj}}{f^{(0)}(0) n_0 n_1 (\bar{\Sigma}_f^{-1} \Sigma_1 \bar{\Sigma}_f^{-1})_{jj}}$, then again the pooling estimator will lose efficiency. 

From the above derivation, we also notice that the asymptotic variance is related to the composite matrices $\bar{\Sigma}$, $\bar{\Sigma}_\sigma$, and $\bar{\Sigma}_f$, which are weighted sums (rather than differences) of the covariance matrices $\Sigma_0$ and $\Sigma_1$. By utilizing the matrix Loewner order, we have that 
\begin{align*}
    & \Sigma_1 \preceq \frac{n_0 + n_1}{n_1} \bar{\Sigma}, \quad \bar{\Sigma}^{-1} \preceq \frac{n_0}{n_0 + n_1} \Sigma_0^{-1} + \frac{n_1}{n_0 + n_1} \Sigma_1^{-1}, \\
    & \Sigma_1 \preceq \frac{n_0 + n_1}{n_1 f^{(1)}(0)} \bar{\Sigma}_f, \quad \bar{\Sigma}_f^{-1} \preceq \frac{n_0 f^{(0)}(0)}{n_0 + n_1} \Sigma_0^{-1} + \frac{n_1 f^{(1)}(0)}{n_0 + n_1} \Sigma_1^{-1}.
\end{align*}
We can then get that 
\begin{align*}
    & (\bar{\Sigma}^{-1})_{jj} \le \frac{n_0}{n_0 + n_1} (\Sigma_0^{-1})_{jj} + \frac{n_1}{n_0 + n_1} (\Sigma_1^{-1})_{jj}, \\
    & (\bar{\Sigma}^{-1} \Sigma_1 \bar{\Sigma}^{-1})_{jj} \le \frac{n_0}{n_1} (\Sigma_0^{-1})_{jj} + (\Sigma_1^{-1})_{jj}, \\
    & (\bar{\Sigma}_f^{-1} \Sigma_1 \bar{\Sigma}_f^{-1})_{jj} \le \frac{n_0 f^{(0)}(0)}{n_1 f^{(1)}(0)} (\Sigma_0^{-1})_{jj} + (\Sigma_1^{-1})_{jj}.
\end{align*}
From these expressions, we can see that for both OLS and QR pooling estimators, the component-wise asymptotic variance can be well controlled by the corresponding diagonal element of matrices $\Sigma_0^{-1}$ and $\Sigma_1^{-1}$, which is not overly sensitive to the difference of covariance matrices $(\Sigma_1 - \Sigma_0)$. Instead, the primary drivers affecting estimation performance are the residual shift, i.e., $\sigma_1^2 / \sigma_0^2$ and $f^{(1)}(0) / f^{(0)}(0)$. Moreover, as the source sample size $n_1$ increases, the condition required to improve estimation efficiency becomes less restrictive. Consequently, our detection rule $n_0 \mathbb{E}[f_i^{(0)}] \le h_2 n_k \mathbb{E}[f_i^{(k)}]$ remains reasonable and robust under covariate shift scenarios.

To further validate our detection rule under strong covariate shift, we conduct additional simulations following the setting of \cite{he2024transfusion}. The covariates of each source study now follow $\mathcal{N} ( \mathbf{0}_p, \boldsymbol{\Sigma}_{\bx} + \boldsymbol{A}_{\bx}^{(k)\top} \boldsymbol{A}_{\bx}^{(k)} )$, where $\boldsymbol{A}_{\bx}^{(k)}$ is a random matrix whose entries are $0.3$ with probability $0.3$ and $0$ with probability $0.7$. The simulation results are shown in Figure \ref{simu:est_homo_normal_cov}, \ref{simu:est_homo_t_cov}, \ref{simu:est_hetero_t_cov}, \ref{simu:infer_normal_cov} and \ref{simu:infer_t_cov}. All figures convey similar conclusions to those in the main text, demonstrating the robustness of our proposed methods under strong covariate shift. Specifically, our method performs almost as well as the oracle benchmark in the parameter estimation, whereas the two-step approach exhibits noticeably higher $\ell_2$-error, underscoring its vulnerability to heterogeneous covariance structures.

\subsubsection{Algorithm of the Constrained L1-minimization Problem}

In Section \ref{section2-1}, we introduce the transfer framework \eqref{l1 transfer}, which can simultaneously tackle the parameter/residual/residual shift when utilizing the information from source studies. Here we introduce the iterative algorithm to solve the constrained $\ell_1$-minimization problem \eqref{l1 transfer} when the oracle transferable set $\bc_h$ is known, which is also applied in \cite{li2023estimation}. We can approximate \eqref{l1 transfer} with a joint $\ell_1$-penalized quantile regression ($\ell_1$-QR) form, leveraging the equivalence between Lasso and Dantzig selectors \citep{bickel2009simultaneous}:
\begin{align} \label{alternative l1 transfer}
	\left(\wbeta, \{\wdelta^{(k)}\}_{k\in \bc_h}\right)= \underset{\bbeta, \{\bdelta^{(k)}\}_{k \in \bc_h}}{\arg \min}
	\left\{\sum_{k \in \{0\} \cup \bc_h} L^{(k)} (\bbeta + \bdelta^{(k)}) + \lambda_{\bbeta} \| \bbeta \|_1 + \sum_{k \in \bc_h} \lambda_k \| \bdelta^{(k)} \|_1 \right\}.
\end{align}
This formulation is similar to the co-training step in \cite{he2024transfusion}, with replacing $\bdelta^{(k)}$ by $\bw^{(k)} - \bbeta$. As suggested by \cite{li2023estimation}, the algorithm for solving \eqref{alternative l1 transfer} proceeds naturally using iterative alternating direction methods. At a high level, our method works as follows. We first obtain an initial estimator $\wbeta_{0}$ by $\ell_1$-QR solely on the target data (or the estimator $\wbeta_{t-1}$ from the last iteration). Then for each source $k \in \bc_h$ we subtract $(\bx_{i}^{(k)})^{\top} \wbeta_{t-1}$ from the response $y_{i}^{(k)}$ and obtain the estimator $\wdelta^{(k)}_{t}$ by $\ell_1$-QR on the refined source dataset. After that, we take $\{\wdelta^{(k)}_{t}\}_{k \in \bc_h}$ into the objective function \eqref{alternative l1 transfer}, solve for the estimator $\wbeta_{t}$, and begin the next round. The iteration stops when estimators converge, where the detailed procedure is summarized in Algorithm \ref{alm:main}.

\begin{algorithm}[ht!] 
\vspace{2mm}
\caption{Parameter estimation with oracle transferable set $\bc_h$}
\label{alm:main}
\KwIn{Target data $\{((\bx_{i}^{(0)})^{\top}, y_{i}^{(0)})_{i=1}^{n_0}\}$, source data $\{((\bx_{i}^{(k)})^{\top}, y_{i}^{(k)})_{i=1}^{n_k}\}_{k \in \bc_h}$, \qquad tuning parameters $\lambda_{\bbeta}, \{ \lambda_k \}_{k \in \{0\} \cup \bc_h}$.}
Compute the initial estimate $\wbeta_{0} \leftarrow \ell_1 \text{-QR } \text{ on } [((\bx_{i}^{(0)})^{\top}, y_{i}^{(0)})_{i=1}^{n_0},\lambda_0]$. \\
\For{$t=1, 2, \dots$}{
    \begin{enumerate}[leftmargin=*]
            \item Compute $\wdelta^{(k)}_{t} \leftarrow \ell_1 \text{-QR } \text{ on } [ ((\bx_{i}^{(k)})^{\top}, y_{i}^{(k)} - (\bx_{i}^{(k)})^{\top} \wbeta_{t-1})_{i=1}^{n_k}, \lambda_k ], k \in \bc_h$.
            \item Solve the pooled $\ell_1$-QR problem
            $$
                \wbeta_t = \underset{\bbeta}{\arg \min} \left\{\sum_{k \in \{0\} \cup \bc_h} L^{(k)} (\bbeta + \wdelta^{(k)}_{t}) + \lambda_{\bbeta} \left\| \bbeta \right\|_1 \right\}.
        $$
            \item When $\wbeta_t$ converges at step $T$, record the present result and exit the iteration. 
    \end{enumerate}
}
\KwOut{$\{\wbeta, \{\wdelta^{(k)}\}_{k \in \bc_h}\}:=\{\wbeta_{T}, \{\wdelta^{(k)}_T\}_{k \in \bc_h}\}$.}
\end{algorithm}

\subsubsection{A Practical Guide for Tuning Parameters}

There are several hyperparameters in our transfer procedures, including detection thresholds $t_1, t_2$, tuning parameters $\{\lambda_k\}_{k=0}^K, \lambda_{\bbeta}$, bandwidth parameters $\{b_k\}_{k=0}^K$ in Algorithm \ref{alm:unify}, $\{\lambda_\tau^{(k)}, \lambda_{\btheta}^{(k)}\}_{k \in \{0\} \cup \wc}, \bar{\lambda}$ in Algorithm \ref{alm:inference}. To attain better performance in the simulations, here we propose several ways to determine these hyperparameters. 

\begin{itemize}[leftmargin=*]
    \item To select the $\ell_1$-QR tuning parameters, we deviate from the vanilla penalized quantile regression and instead follow the $\lambda$-construction idea in \cite{belloni2011l1, giessing2023debiased}. Specifically, the regularization term changes from $\lambda\|\bbeta\|_1$ to $\lambda \sqrt{\tau(1-\tau)} \sum_{j=1}^p \widehat{\sigma}_{k,j} |\beta_j|$, with $\widehat{\sigma}_{k,j}^2 := n_k^{-1}\sum_{i=1}^{n_k}(\bx_{ij}^{(k)})^2$ and $\lambda = c^* \cdot \Lambda_\tau^{(k)}(\alpha^* | \{\bx_{i}^{(k)}\}_{i=1}^{n_k})$,  where $\Lambda_\tau^{(k)}(\alpha^* | \{\bx_{i}^{(k)}\}_{i=1}^{n_k})$ is the $(1 - \alpha^*)$-th quantile of $\Lambda_\tau^{(k)} | \{\bx_{i}^{(k)}\}_{i=1}^{n_k}$ and
    $$
    \Lambda_\tau^{(k)} | \{\bx_{i}^{(k)}\}_{i=1}^{n_k} := \max _{1 \leq j \leq p}\left| \frac{1}{n_k}\sum_{i=1}^{n_k} \frac{\left(\tau-1\left\{U_i \leq \tau\right\}\right) \bx_{ij}^{(k)}}{\sqrt{\tau(1-\tau) \widehat{\sigma}_{k,j}^2 }}\right|,
    $$
    with $U_1, \ldots, U_{n_k}$ are i.i.d. Uniform$(0,1)$ random variables, independent of the covariates $\{\bx_{i}^{(k)}\}_{i=1}^{n_k}$. Such a pilot construction can save the computation cost a lot compared to the cross-validation. $\alpha^*$ is set as 0.05 uniformly, while we set $c^*=1$ for $\lambda_0, \lambda_{\sbeta}$, $c^*=1.5$ for $\{\lambda_k\}_{k=1}^K$, $c^* = 2\sqrt{\tau(1-\tau)}$ for $\{\lambda_{\tau}^{(k)}\}_{k \in \{0\} \cup \wc}$. This procedure can be executed through R functions \texttt{LassoLambdaHat()} in the \texttt{quantreg} package. 
    \item For the bandwidth parameters $\{b_k\}_{k=0}^K$, we follow the Powell bandwidth construction \cite{dai2021inference, huang2023estimation} with estimators $\tbeta, \{\tdelta^{(k)}\}_{k=1}^K$ in Algorithm \ref{alm:unify}. Specifically, 
    \begin{align*}
        & b_k = \left( \Phi^{-1} (\tau + \widetilde{b}_k) - \Phi^{-1} (\tau - \widetilde{b}_k) \right) \min \left\{ \sqrt{\widehat{var}\left[ \widehat{\boldsymbol{e}}^{(k)}_i \right]}, \frac{\widehat{Q}_{0.75} \left[ \widehat{\boldsymbol{e}}^{(k)}_i \right] - \widehat{Q}_{0.25} \left[ \widehat{\boldsymbol{e}}^{(k)}_i \right]}{1.34} \right\}, \\
        & \widetilde{b}_k = n_k^{-1/3} \left\{ \Phi^{-1} (0.975) \right\}^{2/3} \left( \frac{1.5 \phi^2(\Phi^{-1}(\tau))}{2 \Phi^{-2}(\tau) + 1} \right)^{1/3},
    \end{align*}
    where $\widehat{\boldsymbol{e}}^{(k)}_i = y_i^{(k)} - (\bx_i^{(k)})^\top (\tbeta + \tdelta^{(k)})$, $\widehat{var} [\widehat{\boldsymbol{e}}^{(k)}_i]$ and $\widehat{Q}_{\tau} [\widehat{\boldsymbol{e}}^{(k)}_i]$ denote the sample variance and the $\tau$-th sample quantile of $\{\widehat{\boldsymbol{e}}^{(k)}_i\}_{i=1}^{n_k}$.
    \item For the threshold parameter $t_1, t_2$ in Algorithm \ref{alm:unify}, one practical strategy is to conduct the cross-validation procedure \citep{tian2022transfer, zhang2022transfer} for a target study and several candidate source studies, where we evaluate a grid of threshold values and choose the thresholds that yield the best prediction performance (least quantile loss on the validation set). The chosen thresholds can then be applied across all related transfer learning tasks. For moderately complex scenarios, this procedure is not computationally expensive and provides sufficiently accurate threshold choices. To further examine the impact of threshold selection, we conducted sensitivity analyses by varying the values of $t_1$ and $t_2$ in the simulations, as shown in Figures \ref{simu:est_homo_normal_thres1} and \ref{simu:est_homo_normal_thres2}. The results indicate that as long as the thresholds lie within a reasonable interval, our method’s performance remains close to the oracle benchmark, consistent with the detection consistency established in Theorem \ref{thm:detection}. For all simulations in this paper, we set $t_1 = 5$ and $t_2 = 0.3$. 
    \item For the post-selected threshold $\bar{\lambda}$, we directly let it be 0.01 for simplicity.
    \item For the parameters $\lambda_{\btheta}^{(k)}$ in the projection, we use the common pivotal choice $c_{\btheta} \sqrt{n_k \log p}$ in the high-dimensional settings. For different simulations (with settings described in Section \ref{section-sup2}), the constant $c_{\btheta}$ is set as $0.1$ for the target study and source studies under normal distributions without strong covariate shift, $0.05$ for source studies under normal distributions with strong covariate shift, $0.08$ for the target study and source studies under $t$ distributions without strong covariate shift, $0.03$ for source studies under $t$ distributions with strong covariate shift. 
\end{itemize}

\subsubsection{Quality Control for Statistical Inference}

In Section \ref{section4}, we introduce the Neyman orthogonal debiased approach to incorporate informative source studies for statistical inference. When the source studies in the oracle transferable set $\bc_h$ satisfy the requirement for Theorem \ref{thm:inference}, we can get the asymptotic result of Algorithm \ref{alm:inference} by combining Theorem \ref{thm:detection} and \ref{thm:inference}.

\begin{corollary} \label{cor:oracle inference}
    Assume all conditions required by Theorem \ref{thm:convergence}, \ref{thm:detection} and \ref{thm:inference} hold for the oracle transferable set $\bc_h$. For the estimator $\walpha$ obtained from Algorithm \ref{alm:inference}, we have
    \small
	\begin{align*}
		\sigma_{\bc_h}^{-1} \sqrt{n_{\bc_h}} (\walpha - \beta_1^*) = \mathbb{U}_{\bc_h}(\tau) + o_p (1), \qquad \mathbb{U}_{\bc_h}(\tau)  \leadsto \mathcal{N}(0,1),
	\end{align*}
    \normalsize
	where $\sigma_{\bc_h}, \mathbb{U}_{\bc_h}$ follow the definitions of $\sigma_{\bc}, \mathbb{U}_{\bc}$ in Theorem \ref{thm:inference} with $\bc = \bc_h$. Moreover, the estimator $\widehat{\sigma}_{\wc}^{2} = \tau (1 - \tau) n_{\wc} [\sum_{k \in \{0\} \cup \wc} \sum_{i=1}^{n_k} (\wv_{i}^{(k)})^2 ]^{-1}$ is consistent. 
\end{corollary}

Although Corollary 2 establishes the theoretical guarantee for post-detection inference, its validity still requires the stronger conditions imposed in Theorem 4. In particular, the parameter shift must be sufficiently small to control the approximation error $r_i^{(k)} = x_{i,1}^{(k)}\delta_1^{(k)}$, $k \in \bc_h$. This error cannot be eliminated by the orthogonalization approach, since it involves the same covariate $x_{i,1}^{(k)}$ as the target parameter $\beta^*_1$. Following the orthogonal inference framework of \cite{belloni2019valid}, these approximation errors must be sufficiently small to ensure valid asymptotic inference. Specifically, Condition \ref{cond:moment conditions} requires the effects of these approximation errors to be negligible at the pooled-sample inference scale, where a sufficient similarity requirement is $\max_{k\in\bc_h}|\delta_1^{(k)}| = o(n_{\bc_h}^{-1/2})$. This requirement is understandable since we need smaller parameter contrast for better normality results. 

However, verifying such strong similarity conditions is challenging in practice, especially when the source set $\wc$ is selected in a data-driven manner. To address this concern, we propose a practical quality-control procedure for joint debiased inference. Specifically, we can use prior knowledge or estimated parameter contrasts $\widetilde{\delta}_1^{(k)}$ to exclude source studies with potentially non-negligible parameter shifts, and then perform a bootstrap-based normality test on the debiased estimator. If the normality hypothesis is rejected, we revert to target-only debiasing. The full procedure is summarized in Algorithm \ref{alm:quality control}. We emphasize that this procedure serves as a practical safeguard rather than a theoretical guarantee. In particular, it does not formally verify Conditions \ref{cond:moment conditions}--\ref{cond:density estimation} or guarantee valid confidence interval coverage, even when the normality hypothesis is not rejected.

\begin{algorithm}[ht] 
\caption{Quality control for Algorithm \ref{alm:inference}}
\label{alm:quality control}
\KwIn{Same as Algorithm \ref{alm:inference}.}
\begin{enumerate}[label=\Alph*)]
    \item Pre-control: for each source study $k \in \wc$, if prior knowledge suggests that the parameter contrast $\delta_1^{(k)}$ is non-negligible, or $|\widetilde{\delta}_1^{(k)}|$ is larger than a prespecified constant, then we remove this study. Denote the new detected set by $\widetilde{\bc}$.
    \item Post-control: normality test with quantile bootstrap.
    \begin{enumerate}[label=(\arabic*), leftmargin=*]
        \item Conduct the quantile pair bootstrap \citep{hahn1995bootstrapping} $R$ times on the data $\{((\bx_{i}^{(k)})^{\top}, y_{i}^{(k)})_{i=1}^{n_k}\}_{k \in \{0\} \cup \widetilde{\bc}}$ to obtain the debiased estimates $\{\walpha_r\}_{r=1}^R$.
        \item Check the normality of $\{\walpha_r\}_{r=1}^R$ by Shapiro-Wilk test \citep{shapiro1965analysis}. \\
        - If rejecting null hypothesis ($p \text{-value} < 0.05$), we choose to only debias on the target study by $\arg \min_ {\alpha \in \mathcal{A}_{0}} | \sum_{i=1}^{n_0} \wpsi_i^{(0)}(\alpha) |$, where the search region $\mathcal{A}_{0} = \{ \alpha \in \mathbb{R}:| \alpha - \widehat{\beta}_1| \le 10 [ n_0^{-1} \sum_{i=1}^{n_0} (x^{(0)}_{i,1})^2  ]^{-1/2} / \log n_0 \}$, which is exactly the approach in \cite{belloni2019valid}. \\
        - If not rejecting null hypothesis, conduct Algorithm \ref{alm:inference} with $\widetilde{\bc}$.
    \end{enumerate}
\end{enumerate}
\KwOut{Debiased estimate $\walpha$.}
\end{algorithm}

In Section \ref{section5.2}, we have shown that transfer debiased estimators can enjoy sharper inference results if the contrast level $\delta_1^{(k)}$ is small. Now with other settings unchanged, we let $|\delta_1^{(k)}|=0.1$ and still compare the methods conducted in \ref{section5.2}, where the results are shown in Table \ref{tab:infer_normal_h} and Figure \ref{simu:infer_normal_h}. 

From Table \ref{tab:infer_normal_h} and Figure \ref{simu:infer_normal_h}, when $|\delta_1^{(k)}|$ is large, unfortunately our \emph{Debias\_trans} performs much worse than the target debiasing methods, which is in our expectation since Condition \ref{cond:moment conditions} is no longer satisfied. The density plots in Figure \ref{simu:infer_normal_h} show that the transfer debiased estimate is no longer normal, and the Shapiro-Wilk test rejects the normality hypothesis. In a similar way, the bootstrap approach in Algorithm \ref{alm:quality control} can also indicate us that debiasing on source studies will generate unreliable estimates. Therefore, we recommend only debiasing on the target study, namely the \emph{Debias\_trans\_tar} approach, which can keep the coverage proportions with more stable interval lengths than \emph{Debias\_dl}.

\clearpage

\subsection{Additional Simulation Results} \label{section-sup2}

We first introduce the simulation details for the motivating example in Section \ref{section1}. The setting is similar to the simulation in \cite{zhang2022transfer}. We set $p=500, s=10, n_0=200, \sbeta=(\mathbf{0.5}_s, \mathbf{0}_{p-s})^\top, \bx_i^{(0)} \sim \mathcal{N}(\mathbf{0}_p, \boldsymbol{\Sigma}_{\bx})$ with $\boldsymbol{\Sigma}_{\bx}=(0.7^{|i-j|})_{1 \leq i, j \leq p}$. We consider only one source study with sample size $n_1$ varying from 100 to 500, and we assume no parameter or covariate shift, i.e., $\bdelta^{(1)} = 0, \bx_i^{(1)} \sim \mathcal{N}(\mathbf{0}_p, \boldsymbol{\Sigma}_{\bx})$.  The key change lies in the residual distributions, where $\epsilon_i^{(0)}$ is standard normal, but $\epsilon_i^{(1)}$ follows one of the four types of distributions: 1) standard normal: $\mathcal{N}(0,1)$; 2) Cauchy distribution: $\mathcal{C}(0,3)$; 3) mixed Gaussian: $z\mathcal{N}(-3,0.5) + (1-z)\mathcal{N}(3,0.5)$, $z \sim \mathcal{B}ernoulli(\tau)$; 4) noisy normal: $\mathcal{N}(0,5^2)$. It is worth noting that in this setting, the source study is included in the oracle transferable set by the previous definition \citep{li2022transfer, tian2022transfer}. We compare four methods for estimating the target parameter $\sbeta$ at the $0.2$-th quantile, which are: 1) \emph{L1-QR}: $\ell_1$-penalized quantile regression \citep{belloni2011l1} on the target data; 2) \emph{Pooling}: $\ell_1$-penalized quantile regression on the entire dataset; 3) \emph{TransQR\_2step}: two-step transfer for quantile regression \citep{huang2023estimation}; 4) \emph{TransSQR\_2step}: two-step transfer for smoothed quantile regression \citep{zhang2022transfer}. We repeat these methods 50 times and report the average of $\ell_2$-estimation errors in Figure \ref{fig:counter example}.

We conduct extensive experiments for parameter estimation and statistical inference under several circumstances. In particular, we consider the cases with strong covariate shift and $t$-distribution. Strong covariate shift means the covariates of each source study now follow $\mathcal{N} ( \mathbf{0}_p, \boldsymbol{\Sigma}_{\bx} + \boldsymbol{A}_{\bx}^{(k)\top} \boldsymbol{A}_{\bx}^{(k)} )$, where $\boldsymbol{A}_{\bx}^{(k)}$ is a random matrix whose entries are $0.3$ with probability $0.3$ and $0$ with probability $0.7$. $t$-distribution means we set $\epsilon_i^{(k)}, k \in \{0\} \cup \bc_{h_2}$ follow the shifted $t(3)$ distribution. Except the changes we mentioned, all other simulation settings are kept the same as those used in Figure \ref{simu:est_homo_normal} and \ref{simu:infer_normal} of the main manuscript.

\noindent \textbf{Parameter estimation}
    
\begin{itemize}
    \item Figure \ref{simu:est_hetero_normal}: heteroscedastic model, $K=5$, $h_1=5$.
    \item Figure \ref{simu:est_homo_normal_k=20}: homoscedastic model, $K=20$, $h_1=5$.
    \item Figure \ref{simu:est_hetero_normal_k=20}: heteroscedastic model, $K=20$, $h_1=5$.
    \item Figure \ref{simu:est_homo_normal_h=15}: homoscedastic model, $K=5$, $h_1=15$.
    \item Figure \ref{simu:est_hetero_normal_h=15}: heteroscedastic model, $K=5$, $h_1=15$.
    \item Figure \ref{simu:est_homo_normal_cov}: homoscedastic model, $K=5$, $h_1=5$, strong covariate shift.
    \item Figure \ref{simu:est_homo_t_cov}: homoscedastic model, $K=5$, $h_1=5$, strong covariate shift, $t$-distribution.
    \item Figure \ref{simu:est_hetero_t_cov}: heteroscedastic model, $K=5$, $h_1=5$, strong covariate shift, $t$-distribution.
    \item Figure \ref{simu:est_homo_normal_thres1}: homoscedastic model with varying $t_1$ values for \emph{TransQR}. We set $\bc_{h_1}=\{1,2\}, \bc_{h_2}=\{1,2,3,4,5\}$. We use five-fold cross-validation on the target study to compute the prediction error (quantile loss). The results are averaged over $20$ repetitions. 
    \item Figure \ref{simu:est_homo_normal_thres2}: homoscedastic model with varying $t_2$ values for \emph{TransQR}. We set $\bc_{h_1}=\{1,2,3,4,5\}, \bc_{h_2}=\{1,2\}$. We use five-fold cross-validation on the target study to compute the prediction error (quantile loss). The results are averaged over $20$ repetitions. 
\end{itemize}

\noindent \textbf{Statistical inference}

\begin{itemize}
    \item Figure \ref{simu:infer_normal_h} and Table \ref{tab:infer_normal_h}: $|\delta_1^{(k)}| = 0.1$.
    \item Figure \ref{simu:infer_normal_cov} and Table \ref{tab:infer_normal_cov}: strong covariate shift.
    \item Figure \ref{simu:infer_t_cov} and Table \ref{tab:infer_t_cov}: strong covariate shift with $t$-distributed errors.
\end{itemize}

\begin{figure}[p] 
    \begin{center}
        \includegraphics[width=0.9\linewidth]{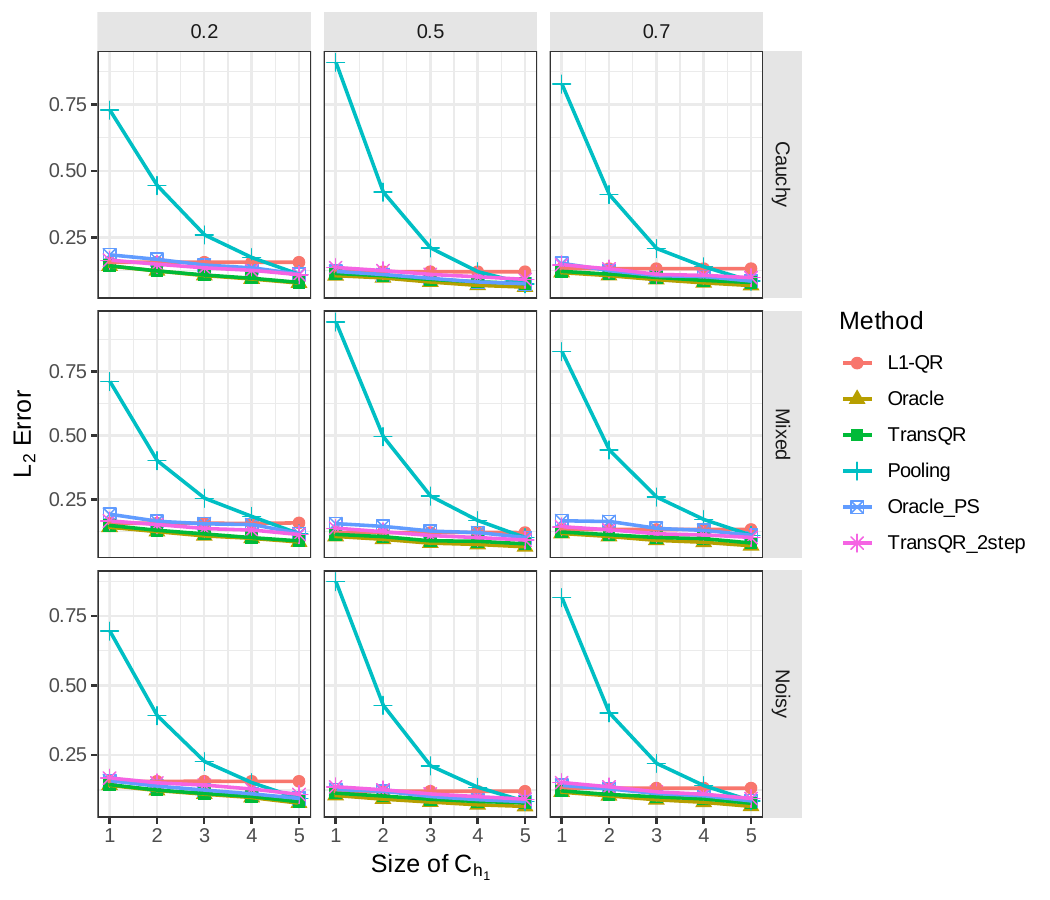}
    \end{center}
    \caption{Average $\ell_2$-error of various methods under heteroscedastic model with $K=5$.}
    \label{simu:est_hetero_normal}
\end{figure}

\begin{figure}[p] 
    \begin{center}
        \includegraphics[width=0.9\linewidth]{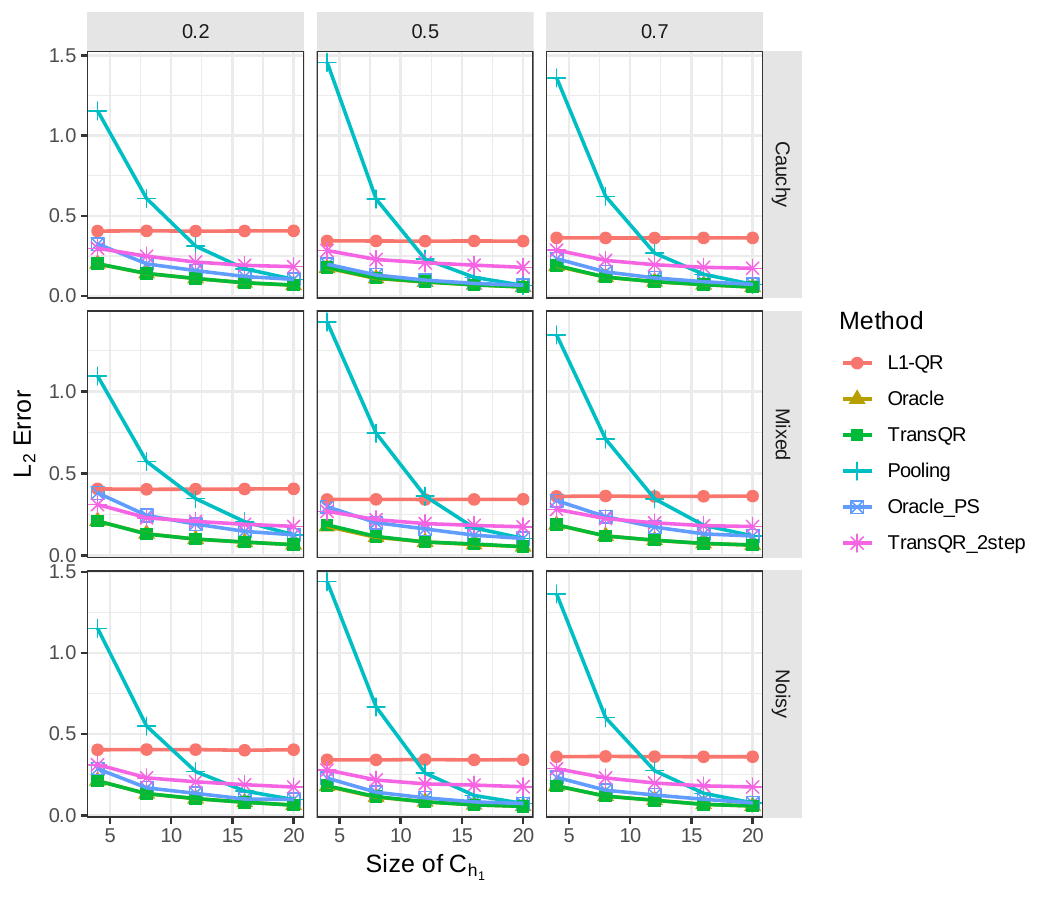}
    \end{center}
    \caption{Average $\ell_2$-error of various methods under homoscedastic model with $K=20$.}
    \label{simu:est_homo_normal_k=20}
\end{figure}

\begin{figure}[p] 
	\begin{center}
		\includegraphics[width=0.9\linewidth]{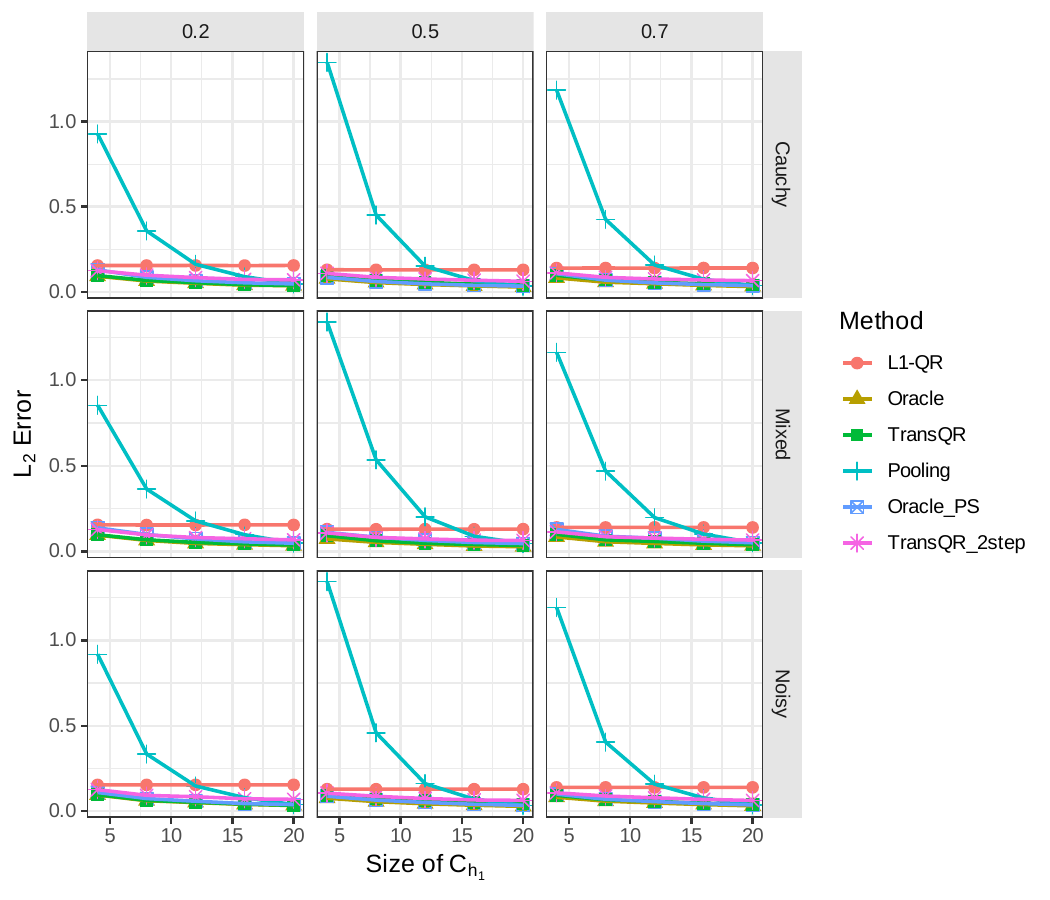}
	\end{center}
	\caption{Average $\ell_2$-error of various methods under heteroscedastic model with $K=20$.}
    \label{simu:est_hetero_normal_k=20}
\end{figure}

\begin{figure}[p] 
	\begin{center}
		\includegraphics[width=0.9\linewidth]{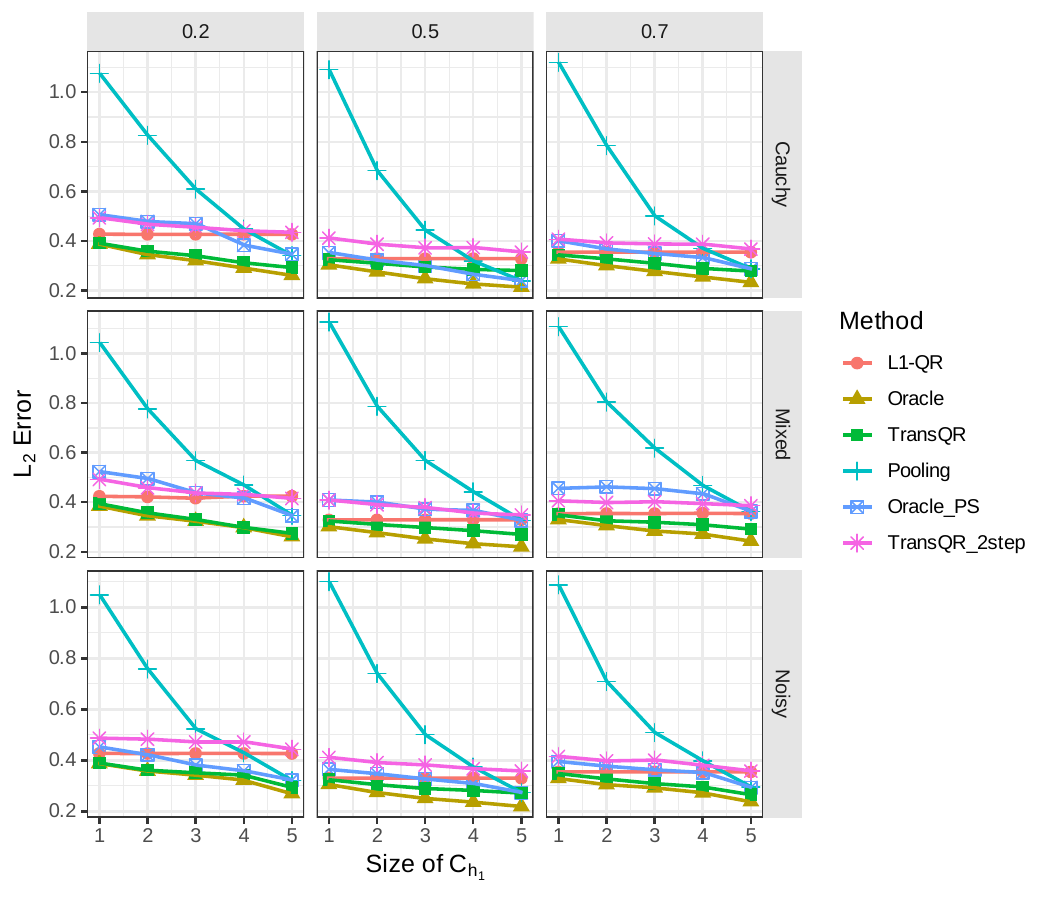}
	\end{center}
	\caption{Average $\ell_2$-error of various methods under homoscedastic model with $h_1=15$.}
    \label{simu:est_homo_normal_h=15}
\end{figure}

\begin{figure}[p] 
	\begin{center}
		\includegraphics[width=0.9\linewidth]{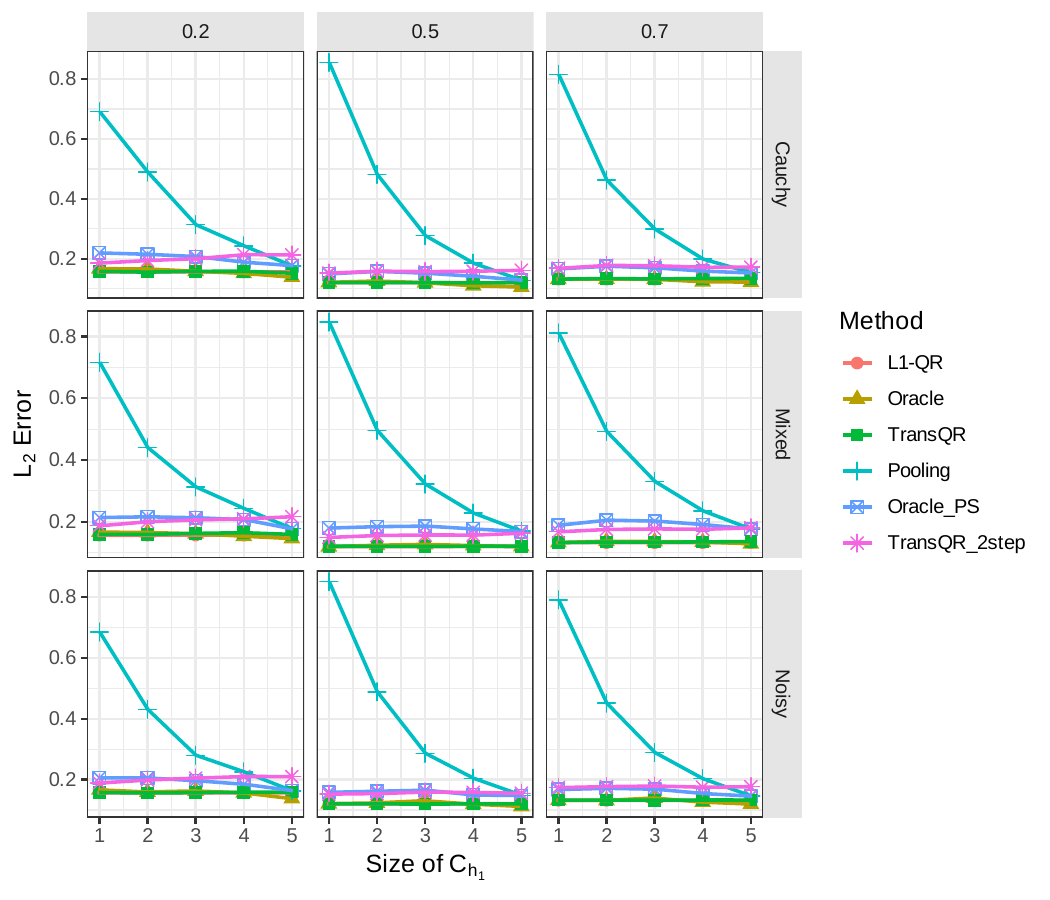}
	\end{center}
	\caption{Average $\ell_2$-error of various methods under heteroscedastic model with $h_1=15$.}
    \label{simu:est_hetero_normal_h=15}
\end{figure}

\begin{figure}[p] 
    \begin{center}
        \includegraphics[width=0.9\linewidth]{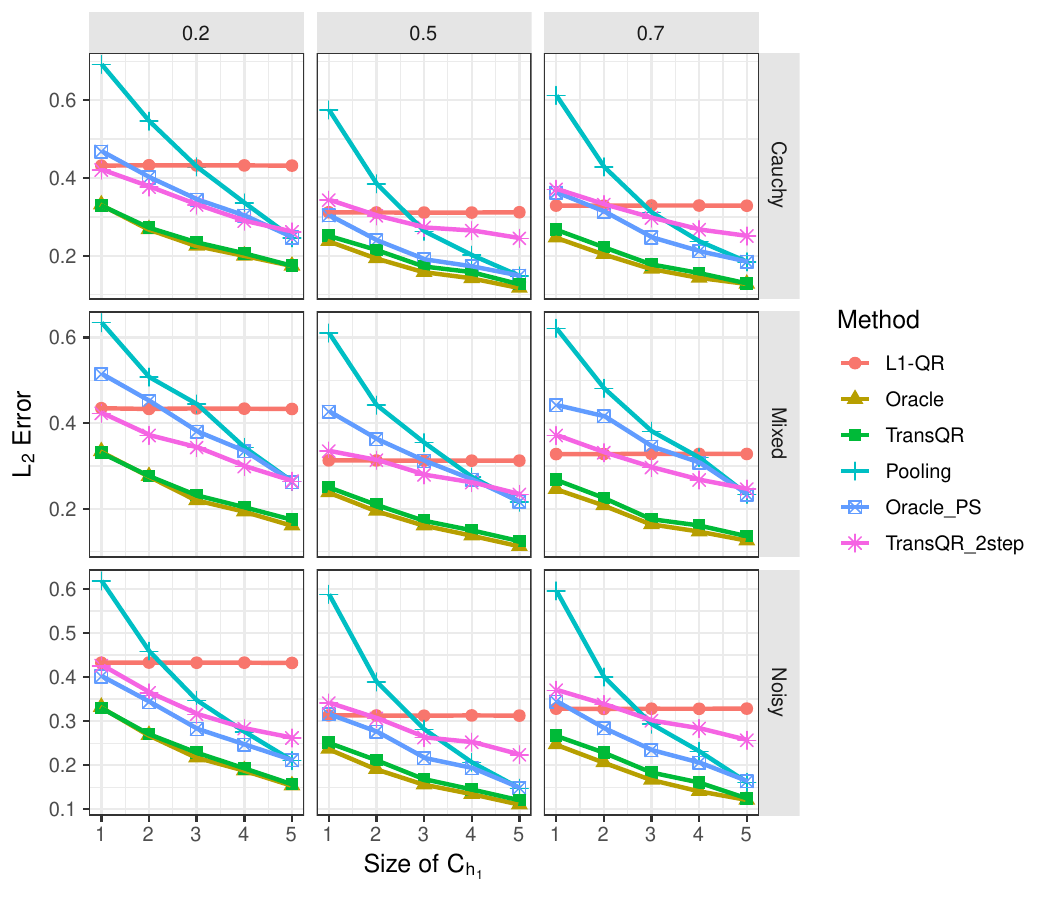}
    \end{center}
    \caption{Average $\ell_2$-error of various methods under homoscedastic model with strong covariate shift.}
    \label{simu:est_homo_normal_cov}
\end{figure}

\begin{figure}[p] 
    \begin{center}
        \includegraphics[width=0.9\linewidth]{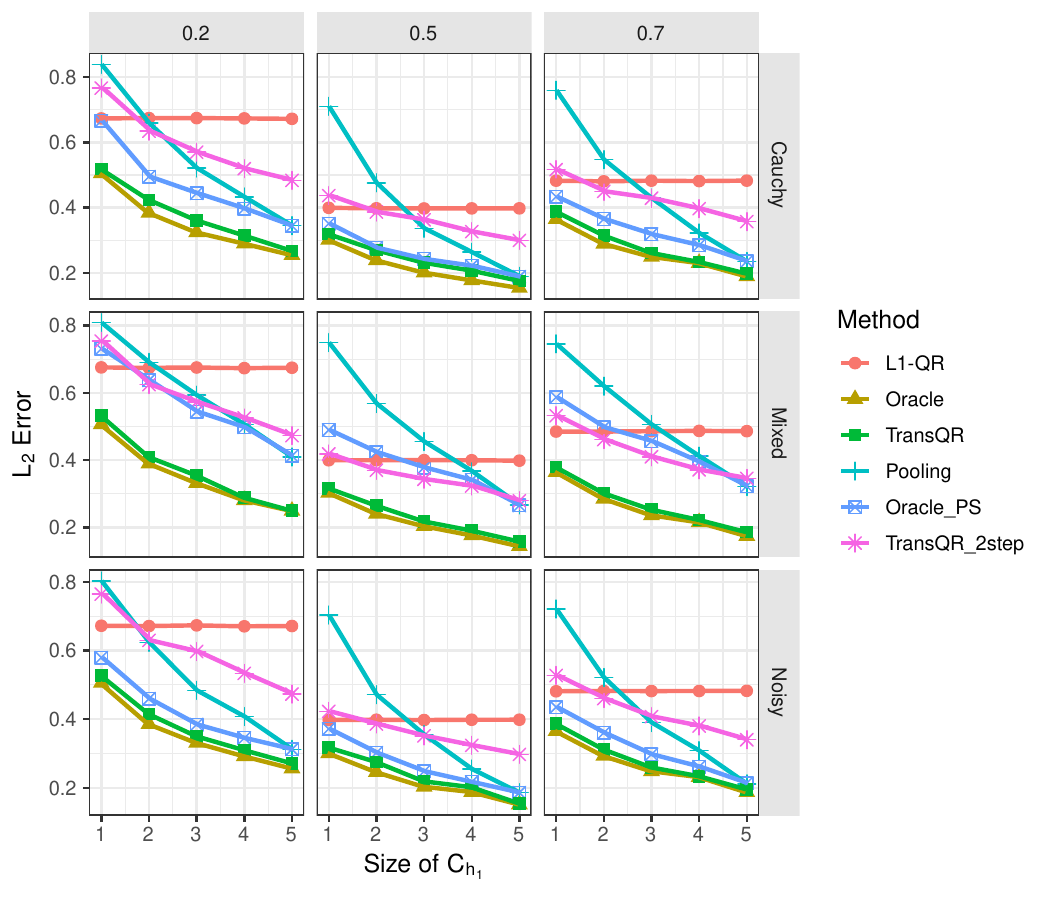}
    \end{center}
    \caption{Average $\ell_2$-error of various methods under homoscedastic model and $t$-distribution with strong covariate shift.}
    \label{simu:est_homo_t_cov}
\end{figure}

\begin{figure}[p] 
    \begin{center}
        \includegraphics[width=0.9\linewidth]{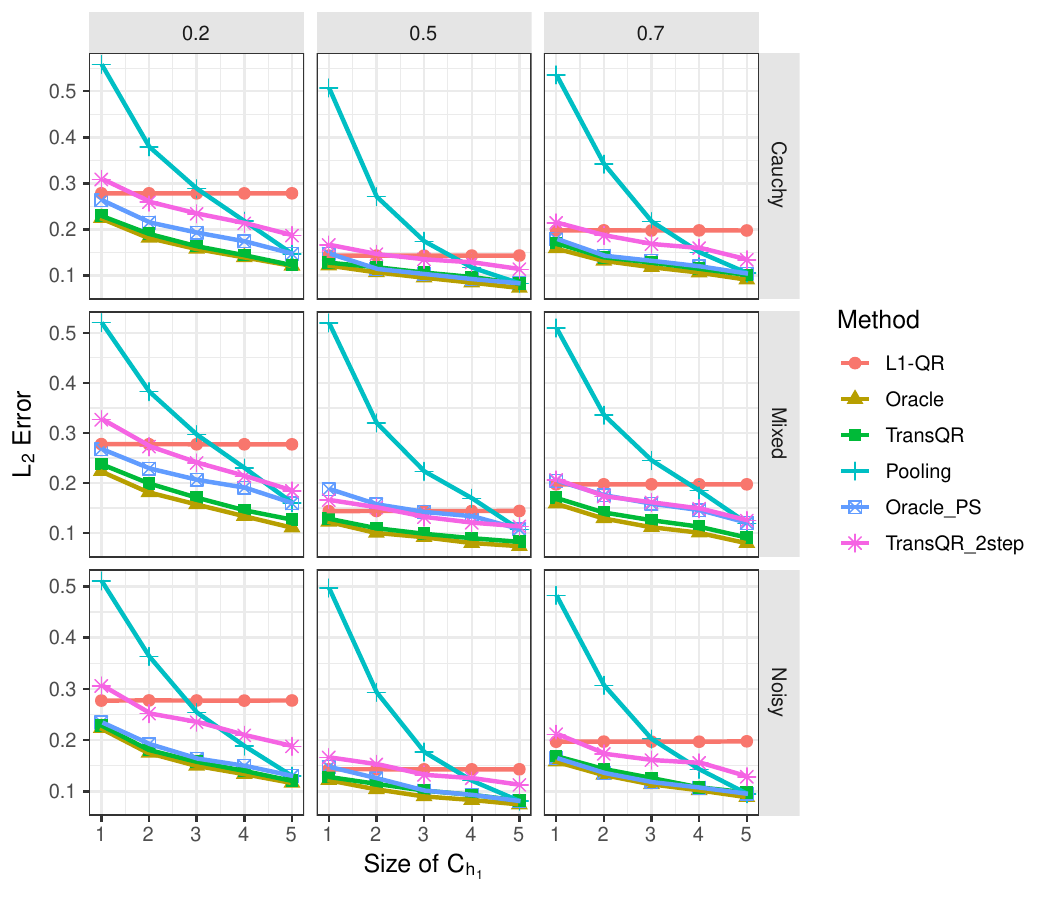}
    \end{center}
    \caption{Average $\ell_2$-error of various methods under heteroscedastic model and $t$-distribution with strong covariate shift.}
    \label{simu:est_hetero_t_cov}
\end{figure}

\begin{figure}[ht]
    \begin{center}
        \includegraphics[width=0.9\linewidth]{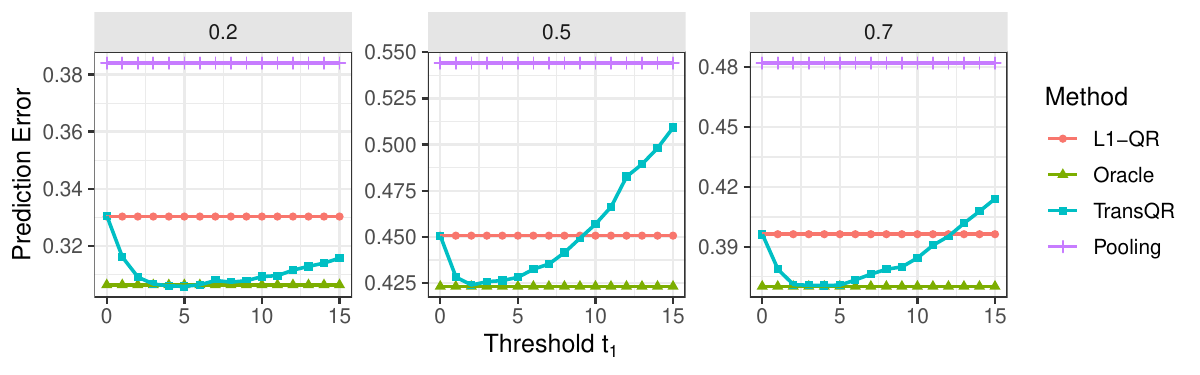}
    \end{center}
    \caption{Average prediction errors under different values of the threshold $t_1$ in \emph{TransQR}.}
    \label{simu:est_homo_normal_thres1}
\end{figure}

\begin{figure}[ht]
    \begin{center}
        \includegraphics[width=0.9\linewidth]{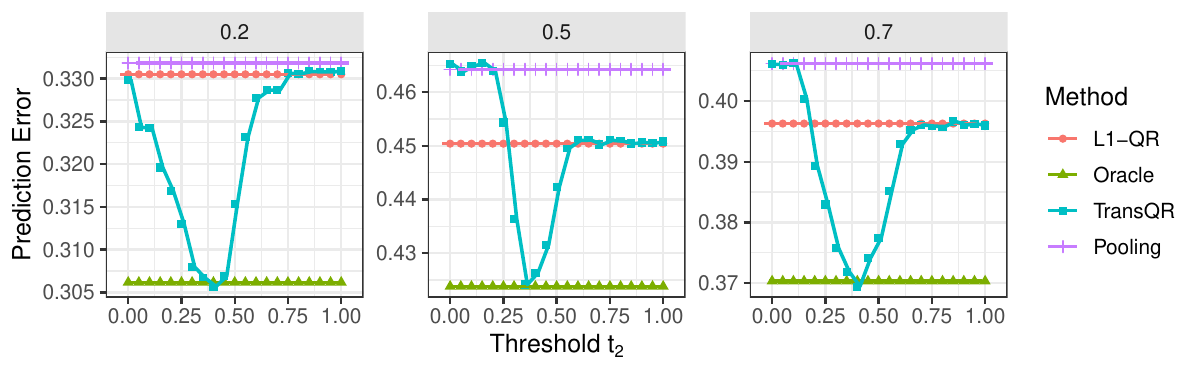}
    \end{center}
    \caption{Average prediction errors under different values of the threshold $t_2$ in \emph{TransQR}.}
    \label{simu:est_homo_normal_thres2}
\end{figure}

\begin{table}[ht]
\footnotesize
\centering
\renewcommand{\arraystretch}{0.8}
\setlength{\tabcolsep}{2.5pt}
\captionsetup{skip=3pt}
\caption{Coverage probability of $95\%$ confidence intervals (CP), average absolute bias (Bias), standard error (SE), and average estimated standard error (ESE) of various methods with $|\delta^{(k)}_1|=0.1$.}
\label{tab:infer_normal_h}
\begin{tabular}{@{}l *{12}{c} @{}}
    \toprule
    \multirow{2}{*}{\normalsize{Method}} & \multicolumn{4}{c}{$\tau=0.2$} & \multicolumn{4}{c}{$\tau=0.5$} & \multicolumn{4}{c}{$\tau=0.7$} \\
    \cmidrule(lr){2-5} \cmidrule(lr){6-9} \cmidrule(l){10-13}
    & CP & Bias & SE & ESE & CP & Bias & SE & ESE & CP & Bias & SE & ESE \\
    \midrule
    \emph{Debias\_tar} & 0.965 & 0.101 & 0.123 & 0.138 & 0.954 & 0.0857 & 0.108 & 0.110 & 0.953 & 0.0880 & 0.110 & 0.117 \\
    \emph{Debias\_trans\_pool} & 0.809 & 0.0892 & 0.107 & 0.0738 & 0.788 & 0.0792 & 0.0955 & 0.0627 & 0.781 & 0.0823 & 0.0998 & 0.0657 \\
    \emph{Debias\_trans\_tar} & 0.971 & 0.0958 & 0.125 & 0.138 & 0.950 & 0.0844 & 0.105 & 0.110 & 0.954 & 0.0875 & 0.110 & 0.117 \\
    \emph{Debias\_trans} & 0.866 & 0.0861 & 0.104 & 0.0807 & 0.836 & 0.0772 & 0.0937 & 0.0677 & 0.853 & 0.0791 & 0.0953 & 0.0711 \\
    \emph{Debias\_dl} & 0.928 & 0.0876 & 0.109 & 0.104 & 0.935 & 0.0714 & 0.0893 & 0.0854 & 0.936 & 0.0782 & 0.0964 & 0.0910 \\
    \bottomrule
\end{tabular}
\end{table}

\begin{figure}[p] 
\begin{center}
    \includegraphics[width=0.9\linewidth]{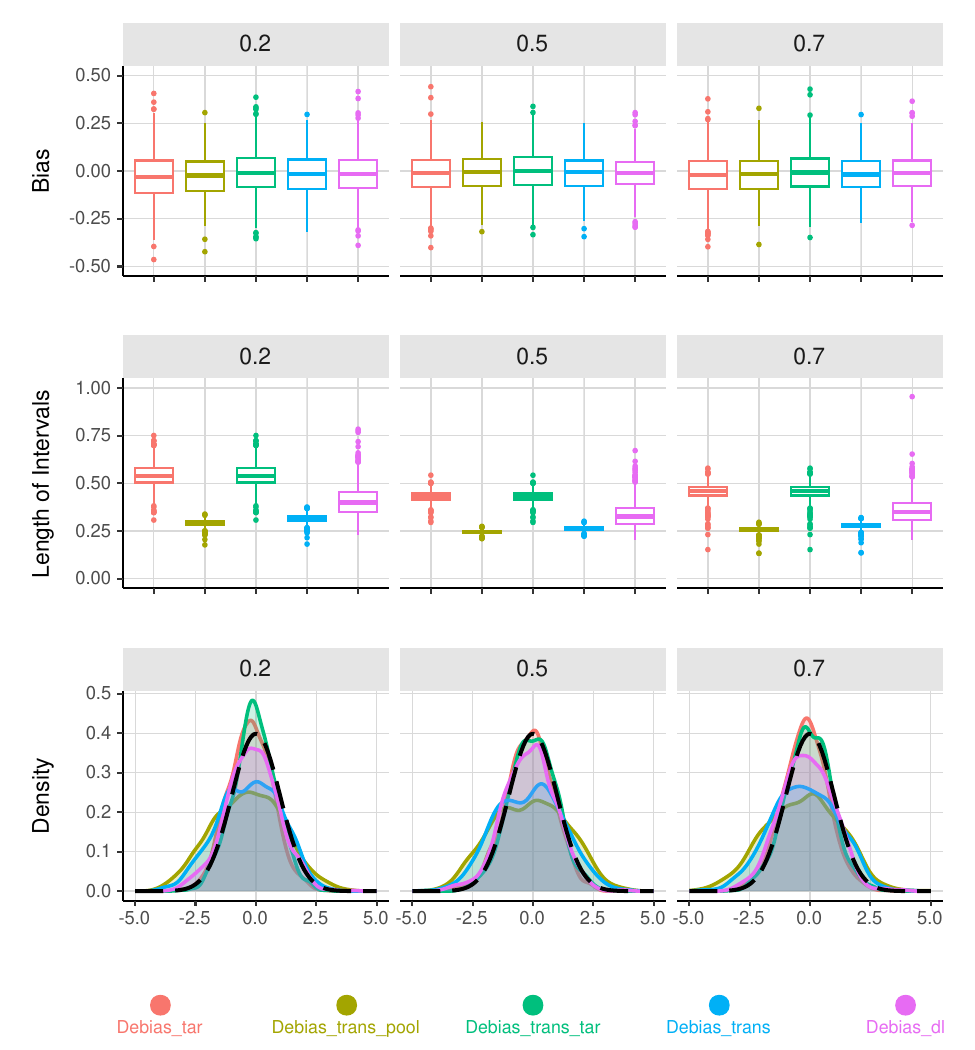}
\end{center}
\caption{Inference results of various methods with $|\delta^{(k)}_1|=0.1$, including boxplots of bias, boxplots of confidence interval lengths, and density plots of normalized estimates compared to $\mathcal{N}(0,1)$ with black dotted line.}
\label{simu:infer_normal_h}
\end{figure} 

\begin{table}[ht]
\footnotesize
\centering
\renewcommand{\arraystretch}{0.8}
\setlength{\tabcolsep}{2.5pt}
\captionsetup{skip=3pt}
\caption{Coverage probability of $95\%$ confidence intervals (CP), average absolute bias (Bias), standard error (SE), and average estimated standard error (ESE) of various methods with strong covariate shift.}
\label{tab:infer_normal_cov}
\begin{tabular}{@{}l *{12}{c} @{}}
    \toprule
    \multirow{2}{*}{\normalsize{Method}} & \multicolumn{4}{c}{$\tau=0.2$} & \multicolumn{4}{c}{$\tau=0.5$} & \multicolumn{4}{c}{$\tau=0.7$} \\
    \cmidrule(lr){2-5} \cmidrule(lr){6-9} \cmidrule(l){10-13}
    & CP & Bias & SE & ESE & CP & Bias & SE & ESE & CP & Bias & SE & ESE \\
    \midrule
    \emph{Debias\_tar} & 0.963 & 0.103 & 0.137 & 0.138 & 0.941 & 0.0853 & 0.108 & 0.110 & 0.954 & 0.0934 & 0.114 & 0.118 \\
    \emph{Debias\_trans\_pool} & 0.927 & 0.0620 & 0.0778 & 0.0696 & 0.944 & 0.0556 & 0.0687 & 0.0656 & 0.918 & 0.0594 & 0.0774 & 0.0655 \\
    \emph{Debias\_trans\_tar} & 0.979 & 0.0964 & 0.134 & 0.138 & 0.952 & 0.0826 & 0.105 & 0.110 & 0.963 & 0.0879 & 0.110 & 0.118 \\
    \emph{Debias\_trans} & 0.961 & 0.0592 & 0.0755 & 0.0769 & 0.963 & 0.0515 & 0.0652 & 0.0710 & 0.948 & 0.0565 & 0.0784 & 0.0713 \\
    \emph{Debias\_dl} & 0.966 & 0.0784 & 0.0986 & 0.110 & 0.945 & 0.0677 & 0.0868 & 0.0884 & 0.951 & 0.0724 & 0.0914 & 0.0962 \\
    \bottomrule
\end{tabular}
\end{table}

\begin{figure}[p] 
\begin{center}
    \includegraphics[width=0.9\linewidth]{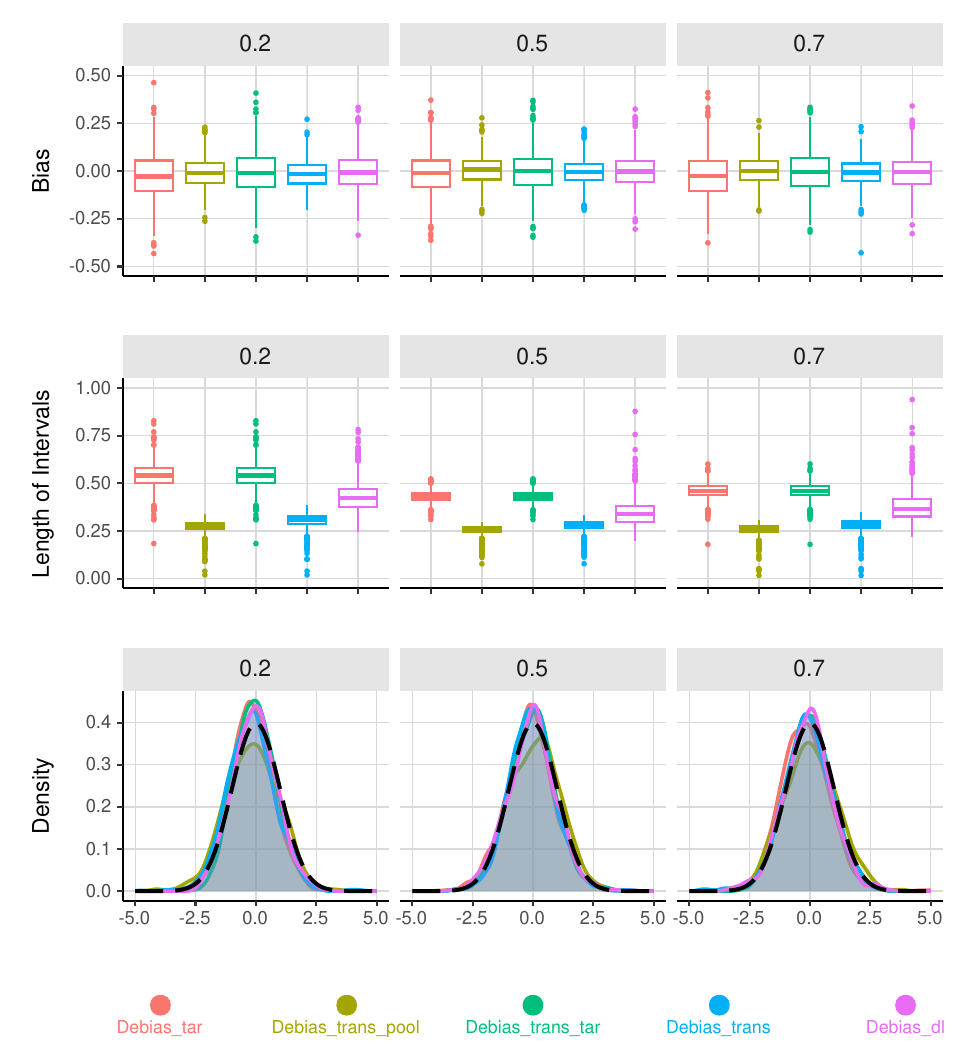}
\end{center}
\caption{Inference results of various methods with strong covariate shift, including boxplots of bias, boxplots of confidence interval lengths, and density plots of normalized estimates compared to $\mathcal{N}(0,1)$ with black dotted line.}
\label{simu:infer_normal_cov}
\end{figure} 

\begin{table}[ht]
\footnotesize
\centering
\renewcommand{\arraystretch}{0.8}
\setlength{\tabcolsep}{2.5pt}
\captionsetup{skip=3pt}
\caption{Coverage probability of $95\%$ confidence intervals (CP), average absolute bias (Bias), standard error (SE), and average estimated standard error (ESE) of various methods under $t$-distribution with strong covariate shift.}
\label{tab:infer_t_cov}
\begin{tabular}{@{}l *{12}{c} @{}}
    \toprule
    \multirow{2}{*}{\normalsize{Method}} & \multicolumn{4}{c}{$\tau=0.2$} & \multicolumn{4}{c}{$\tau=0.5$} & \multicolumn{4}{c}{$\tau=0.7$} \\
    \cmidrule(lr){2-5} \cmidrule(lr){6-9} \cmidrule(l){10-13}
    & CP & Bias & SE & ESE & CP & Bias & SE & ESE & CP & Bias & SE & ESE \\
    \midrule
    \emph{Debias\_tar} & 0.929 & 0.140 & 0.166 & 0.165 & 0.957 & 0.100 & 0.124 & 0.131 & 0.926 & 0.116 & 0.162 & 0.137 \\
    \emph{Debias\_trans\_pool} & 0.905 & 0.0795 & 0.107 & 0.0833 & 0.933 & 0.0624 & 0.0802 & 0.0801 & 0.906 & 0.0689 & 0.0893 & 0.0770 \\
    \emph{Debias\_trans\_tar} & 0.945 & 0.128 & 0.161 & 0.165 & 0.966 & 0.0943 & 0.119 & 0.131 & 0.946 & 0.107 & 0.135 & 0.137 \\
    \emph{Debias\_trans} & 0.944 & 0.0753 & 0.105 & 0.0935 & 0.978 & 0.0588 & 0.0745 & 0.0874 & 0.948 & 0.0648 & 0.0836 & 0.0844 \\
    \emph{Debias\_dl} & 0.940 & 0.110 & 0.135 & 0.138 & 0.956 & 0.0787 & 0.0996 & 0.101 & 0.952 & 0.0857 & 0.108 & 0.113 \\
    \bottomrule
\end{tabular}
\end{table}

\begin{figure}[p] 
\begin{center}
    \includegraphics[width=0.9\linewidth]{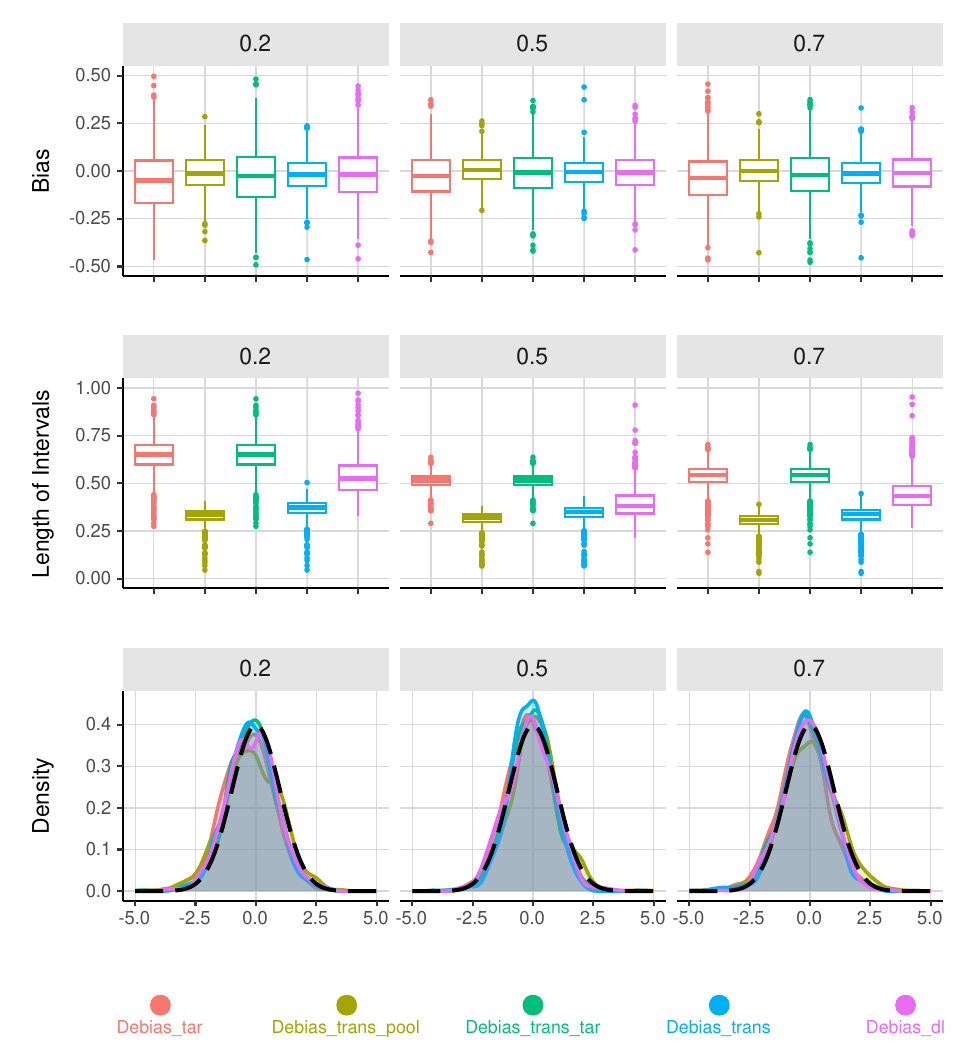}
\end{center}
\caption{Inference results of various methods under $t$-distribution with strong covariate shift, including boxplots of bias, boxplots of confidence interval lengths, and density plots of normalized estimates compared to $\mathcal{N}(0,1)$ with black dotted line.}
\label{simu:infer_t_cov}
\end{figure} 

\clearpage

\subsection{Details of Real Data Application} \label{section-sup3}

We first give a brief introduction for the two genes studied in Section \ref{section6}.

\begin{itemize}[leftmargin=*]
    \item JAM2 (Junctional Adhesion Molecule 2) is a gene that encodes a protein involved in tight junctions between epithelial and endothelial cells, playing a crucial role in cell adhesion and maintaining barrier functions. It is associated with various cellular processes such as cell proliferation and migration. Studies have shown that low expression of JAM2 is linked to poor prognosis in several cancers, including lung adenocarcinoma (LUAD), where reduced JAM2 levels correlate with tumor progression and metastasis. Investigating its low expression helps in understanding cancer progression and identifying therapeutic targets for limiting tumor growth \citep{dong2024jam2}.
    \item SH2D2A encodes the T-cell-specific adapter protein (TSAd), which is essential for T-cell signaling and activation. The expression of SH2D2A is tightly regulated at both the transcriptional and translational levels, with cAMP signaling shown to induce its mRNA expression in T cells. However, low levels of SH2D2A protein may affect immune responses. Notably, reduced SH2D2A expression can enhance T-cell-mediated anti-tumor immunity, making it an interesting target in cancer immunotherapy research \citep{berge2012sh2d2a}.
\end{itemize}

There are some additional points in our practical implementation:
\begin{itemize}[leftmargin=*, itemsep=0pt]
    \item The gene expression data is pre-processed by filtering out genes with constant expression levels and applying standard normalization.
    \item We consider 11 brain tissues as separate target studies; their full names are provided in Table \ref{tab:target tissue}. Two additional brain tissues, ``brain hippocampus'' and ``brain substantia nigra'', are excluded because their substantial within-tissue heterogeneity leads to highly unstable prediction performance across different cross-validation splits, precluding reliable statistical comparisons.
    \item We consider 41 non-brain tissues as candidate source studies; their full names are listed in Table \ref{tab:source tissue}. Two tissues, ``fallopian tube'' and ``kidney cortex'', contain fewer than 100 individuals and are therefore removed during preprocessing. As a result, 39 source tissues are ultimately included in the transfer learning analysis.
    \item We performed threshold value searching for the ``brain cortex'' tissue as the target study in real data applications. The five-fold cross-validation relative prediction errors on the “brain cortex” tissue, shown in Figure \ref{fig:realdata thres}, guide our choice of $t_1 = 7$ and $t_2 = 0.8$ for all tasks. The other regularization parameters follow the guidance in Section 8.1.3.
    \item We observe that source data appears less beneficial at $\tau=0.5$ but seems to be quite helpful at $\tau=0.2$. This suggests that the distribution shift between the target and source studies may vary with the quantile level of model responses.
\end{itemize}

\begin{table}[htb]
\centering
\captionsetup{skip=3pt}
\caption{Complete Name for 11 Target Brain Tissues in the Application}
\label{tab:target tissue}
\footnotesize
\begin{tabular}{|>{\raggedright\arraybackslash}p{6.5cm}|p{3cm}|}
    \hline
    \textbf{Tissue} & \textbf{Abbreviation} \\
    \hline
    brain amygdala & Amygdala \\
    \hline
    brain hypothalamus & Hypothalamus \\
    \hline
    brain caudate basal ganglia & C.B.ganglia \\
    \hline
    brain cerebellar hemisphere & C.hemisphere \\
    \hline
    brain cerebellum & Cerebellum \\
    \hline
    brain cortex & Cortex \\
    \hline
    brain frontal cortex ba9 & F.cortex \\
    \hline
    brain anterior cingulate cortex ba24 & A.C.cortex \\
    \hline
    brain nucleus accumbens basal ganglia & N.C.A.B.ganglia \\
    \hline
    brain putamen basal ganglia & P.B.ganglia \\
    \hline
    brain spinal cord cervical c-1 & S.C.cervical \\
    \hline
\end{tabular}
\end{table}

\begin{table}[htb]
\centering
\captionsetup{skip=3pt}
\caption{Complete Name for 41 Candidate Source Tissues in the Application}
\label{tab:source tissue}
\footnotesize
\begin{tabular}{>{\centering\arraybackslash}p{16cm}}
    \hline
    \textbf{Tissue} \\
    \hline
    adipose subcutaneous, adipose visceral omentum, adrenal gland, artery aorta, artery coronary, artery tibial, bladder, breast mammary tissue, cells cultured fibroblasts, cells ebv-transformed lymphocytes, cervix ectocervix, cervix endocervix, colon sigmoid, colon transverse, esophagus gastroesophageal junction, esophagus mucosa, esophagus muscularis, fallopian tube, heart atrial appendage, heart left ventricle, kidney cortex, kidney medulla, liver, lung, minor salivary gland, muscle skeletal, nerve tibial, ovary, pancreas, pituitary, prostate, skin not sun exposed suprapubic, skin sun exposed lower leg, small intestine terminal ileum, spleen, stomach, testis, thyroid, uterus, vagina, whole blood \\
    \hline
\end{tabular}
\end{table}

\begin{figure}[ht] 
\begin{center}
    \includegraphics[width=0.9\linewidth]{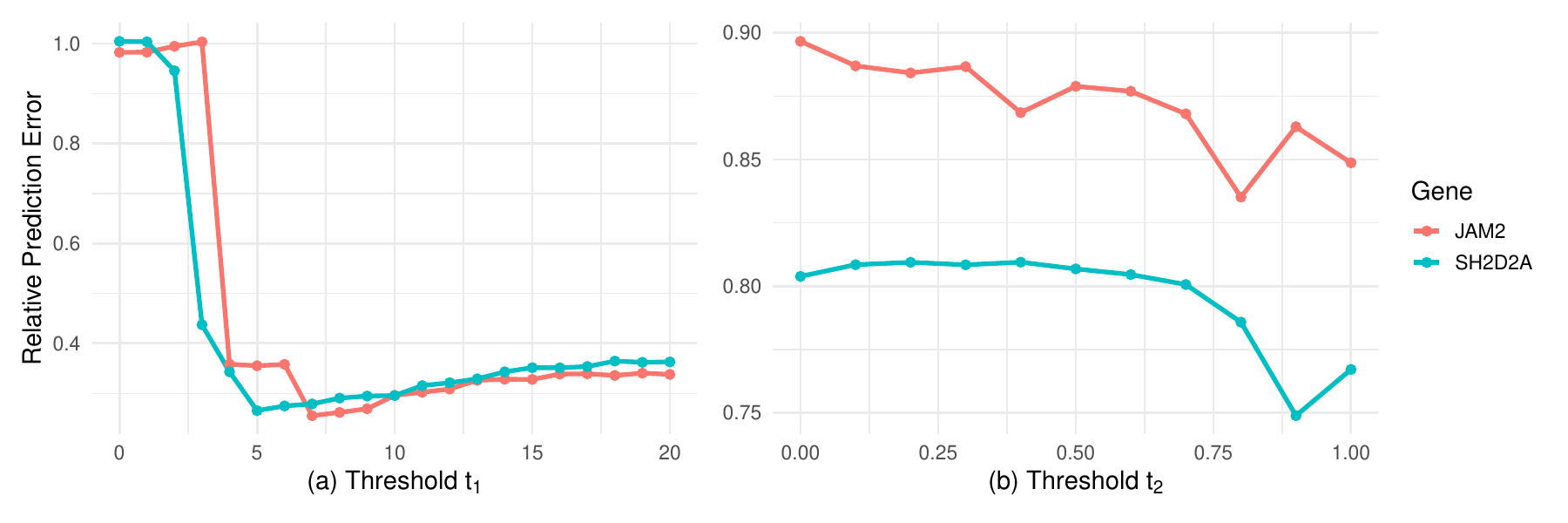}
\end{center}
\caption{Transfer performance under different values for the thresholds $t_1, t_2$ for predicting the cortex tissue at quantile level $\tau = 0.2$.}
\label{fig:realdata thres}
\end{figure} 

\clearpage

\bibliographystyle{apalike}
\setlength{\bibsep}{0pt}
\bibliography{TransQR}

\end{document}